\documentclass[english,a4paper,11pt]{article}

\usepackage[T1]{fontenc}
\usepackage[utf8]{inputenc}
\usepackage{tikz}
\usepackage{algorithm}
\usepackage[noend]{algorithmic}
\usepackage{enumitem}
\usepackage{amsfonts, amssymb, amsmath}
\usepackage{mathtools}
\usepackage{booktabs}
\usepackage{tabularx}
\usepackage{fullpage} 
\usepackage{pdflscape}
\usepackage{multirow}
\usepackage{url}
\usepackage{bbold}
\usepackage{grffile}
\usepackage{appendix}
\usepackage{fullpage}

\usepackage{amsthm}
\newtheorem{theorem}{Theorem}
\newtheorem{lemma}[theorem]{Lemma}
\newtheorem{proposition}[theorem]{Proposition}

\newtheorem{example}[theorem]{Example}

\newcommand{\hide}[1]{}

\newcommand{\N}{\ensuremath{\mathbb{N}}}

  \author{Fran\c{c}ois Cl\`ement, Carola Doerr, Lu\'is Paquete}

  \date{$^1$Sorbonne Universit\'e, CNRS, LIP6, Paris, France\\
  $^2$University of Coimbra, CISUC, Department of Informatics Engineering, Portugal}

\begin{document}
\title{Star Discrepancy Subset Selection: Problem Formulation and Efficient Approaches for Low Dimensions} 
\maketitle


\begin{abstract}
Motivated by applications in instance selection, we 
introduce the \emph{star discrepancy subset selection problem}, which consists of finding a subset of \(m\) out of \(n\) points that minimizes the
star discrepancy. 
First, we show that this problem is NP-hard.
Then, we introduce a mixed integer linear formulation (MILP) and
a combinatorial  branch-and-bound (BB) algorithm for the star discrepancy subset selection problem and we evaluate
both approaches against random subset selection and a greedy construction on
different use-cases in dimension two and three. Our results show that the MILP and BB are efficient in dimension two for large and
small $m/n$ ratio, respectively, and for not too large $n$. However,
the performance of both approaches decays strongly for larger dimensions and set sizes. 

As a side effect of our empirical comparisons we obtain point sets of  
discrepancy values that are much smaller than those of common low-discrepancy sequences, random point sets, and of
Latin Hypercube Sampling. This suggests that subset selection 
could be an interesting approach for generating point sets of small discrepancy value.
\end{abstract}

\sloppy
\section{Introduction}
\label{sec:intro}

Discrepancy measures are metrics designed to quantify how regularly a set of points is distributed in a given space. Several discrepancy notions exist, measuring different aspects of ``regularity''. The arguably most common discrepancy notion is the $L_{\infty}$ star discrepancy. Intuitively speaking, the $L_{\infty}$ star discrepancy of a point set $P \subseteq [0,1]^d$ measures how well the volume of a $d$-dimensional anchored box of the form $[0,q)$ can be approximated by the fraction $|P \cap [0,q)|/|P|$ of points that fall inside this box. More precisely, it measures the largest such deviation between volume and fraction of points. 
Point sets of low $L_{\infty}$ star discrepancy have several important applications, among them Quasi-Monte Carlo integration~\cite{Nie92,DickP10}, one-shot optimization~\cite{CauwetCDLRRTTU20,BousquetGKTV17}, financial mathematics~\cite{GalantiJ97OPtionPricing}, design of experiments~\cite{santner_design_2003}, and many more.  

The design of point sets that guarantee small discrepancy values has been an intensively studied topic in numerical analysis in the last decades, and several constructions are known to achieve a smaller $L_{\infty}$ star discrepancy than randomly sampled points. Among the best-known low-discrepancy constructions are those by Hammersley~\cite{Ham60}, by Sobol'~\cite{Sobol}, and by Halton~\cite{Halton64}. For $d=2$, the construction by Faure~\cite{Faure82} as well as the Fibonacci sequence are often recommended~\cite{Nie92}. 
What is common to all these constructions is that the driving motivation behind their design  are small discrepancy values \emph{in the asymptotic sense,} i.e., when $n=|P| \rightarrow \infty$. While in this setting 
an advantage over random sampling is indeed significant -- the expected $L_{\infty}$ star discrepancy value of i.i.d.~uniformly sampled points is of order $\sqrt{d/n}$~\cite{Doerr14lowerBoundRandomPoints,HNWW01}, whereas the discrepancy of the mentioned low-discrepancy sequences scales as $\ln^{d-1}(n)/n$ -- 
we often require large sample sizes $n$ in order to achieve asymptotic advantage. Low-discrepancy sequences, and in particular Sobol' sequences~\cite{santner_design_2003} are nevertheless often used in practical applications, often competing with uniformly sampled points and Latin Hypercube Sampling (LHS~\cite{LHS}). 

Given the advantageous behavior of point sets of small discrepancy in practice, we study in this work how to choose from a given set $P$ of $n$ points a subset $P_m$ of size $m \le n$ such that the $L_{\infty}$ star discrepancy of $P_m$ is minimized. This \emph{star discrepancy subset selection problem} has its origins in Machine Learning (ML) and in optimization, and in particular in the instance selection problem, where one aims to select from a given set of instances a small subset that maximizes diversity -- with the idea that more diverse instances provide better training opportunities for ML-based approaches. An example for such an approach can be found in~\cite{NeumannGDN018}, where diverse images and instances of the traveling salesperson problem (TSP) are constructed via an evolutionary approach. In each iteration, the evolutionary algorithm generates a set of new instances and a selection operator then updates by selecting instances from the old and the newly generated ones. Since no efficient algorithms were known in~\cite{NeumannGDN018} to address the general star discrepancy subset selection problem, only so-called ``+1'' schemes are considered, which generate only one new instance per iteration. 

Previous results on the NP-hardness of calculating the star discrepancy~\cite{complexity} hint to the difficulty of solving this problem exactly. We show by a reduction from the DOMINATING-SET problem that the decision version of the star discrepancy subset selection problem is NP-hard. We then study the efficiency of algorithmic approaches for the star discrepancy subset selection problem. Simple algorithmic approaches such as random subset selection and iterative greedy selection do not perform well, motivating the design and the analysis of a mixed-integer linear formulation as well as a combinatorial branch-and-bound approach for this problem. 
The mixed-integer linear formulation (MILP) is a natural formulation of the discrepancy subset selection problem that uses a particular property of this problem. In a nutshell, it uses the well-known fact that the worst mismatch between the volume of an anchored box and the fraction of points that fall inside this box is obtained in one of the points that lie on the grid that is spanned by the points in the set. Our branch and bound
(BB) is a classical approach that starts from a greedy solution and uses
combinatorial lower bounds for pruning, which can be computed in an incremental
manner. Our experimental
results for $d=2$ indicate that BB presents better performance for small $m/n$ ratios while MILP performs better 
for large $m/n$ ratios. We relate these findings to the quality of the lower bounds of MILP. Unfortunately, the performance deteriorates strongly already for $n > 140$ and for $d>3$, so that we have to restrict our analysis to the two- and three-dimensional cases.

As a side result, we observe that subset selection can be an interesting approach to generate point sets of small discrepancy values. For our two-dimensional test cases, the star discrepancy of the best found size-$m$ subsets of the Sobol', the Faure, the Halton, and the reverse Halton~\cite{RevHal} sequence is around 50\% smaller than the star discrepancy of the original construction of the same size for $m=20$ and $40$. For larger $m$, the advantage is slightly smaller, but still 40\%, on average, for $m=60$, 36\% for $m=80$, and 44\% for $m=100$. Similar advantages are obtained 
in the three-dimensional case for these four sequences. Much better advantages of at least 60\% are obtained for uniform samples 
for $d=2$ and $d=3$
 and for Latin Hypercubes for $d=3$. For the 
 Fibonacci sequence for $d=2$, in contrast, the advantages are much less important, it is less than 1\% for $m=80$ and $m=100$, but it is slightly above $27\%$ and $22\%$ for $m=20$ and $m=40$, respectively. 

\textbf{Related Work:} The problem of selecting subsets with respect to small  discrepancy values was also the focus of the work on so-called \emph{online thinning},  presented in~\cite{Thinning1}. Online thinning requires a decision maker to either accept or reject a point of a sequence into a selected subset, with the goal of minimizing the discrepancy of the selected set. The process studied in~\cite{Thinning1} assumes, in addition, that at least one out of every two consecutive points has to be selected. The three main differences between their work and ours are:  
(1) while \emph{sequences} are studied in~\cite{Thinning1}, we consider fixed point \emph{sets} $P$ and a fixed target size $m$,  
(2) we optimize over all possible subsets of a given size $m$, and 
(3) in contrast to~\cite{Thinning1}, our approaches are not restricted to uniformly sampled i.i.d.~points.

\textbf{Availability of Code and Results:} The point sets with the best star discrepancy for each value of $m$ in the two-dimensional case that were obtained in our experiments are available at~\url{https://algo.dei.uc.pt/star}. 
This repository is available to the community for reporting improving results and it will be continuously updated for
different values of $m$ and $d$. 
The BB code for $d=2$ is available at \url{https://github.com/luis-paquete/StarDSS}.

\textbf{Structure of the Paper.} We recall in Section~\ref{sec:prelims} relevant definitions and properties of the $L_{\infty}$ star discrepancy problem. In particular, we summarize known bounds, address computational aspects of evaluating the $L_{\infty}$ star discrepancy of a given point set, and briefly introduce the point sets that we consider in the experimental analysis. In Section~\ref{sec:subset} we introduce the discrepancy subset selection problem, prove NP-hardness of its associated decision problem, and discuss some basic properties that are explored by MILP. Our subset selection algorithms are presented in Section~\ref{sec:algos}, while a comparison in terms of running times and solution quality is provided in Section~\ref{sec:comp-algos}. The discrepancy values of the optimal subsets will be discussed in Section~\ref{sec:comp-discrepancy}. We conclude our paper in Section~\ref{sec:conclusions} with a summary of promising directions for future work.

\section{Discrepancy Theory}
\label{sec:prelims}

We briefly summarize in this section relevant background on discrepancy theory. Readers already familiar with this subject can skip this section without loss. Concretely, we first define the discrepancy measure of interest, the $L_{\infty}$ star discrepancy, and summarize known bounds for this measure (Section~\ref{sec:discrepancy-def}). 
Some of the best-known low-discrepancy sequences as well as two random point constructions that will be relevant for our experimental analysis will be presented in Section~\ref{sec:sequence}. In Section~\ref{sec:computation}, we briefly recall a basic property of the star discrepancy evaluation problem, which reduces it to a discrete optimization problem.  

\subsection{The $L_{\infty}$ Star Discrepancy and Known Bounds}
\label{sec:discrepancy-def}

The $L_\infty$ \emph{star discrepancy} $d_{\infty}^*(P)$ of a
point set $P \subseteq [0,1]^d$ 
is defined  as 
\begin{equation}
	d_{\infty}^*(P) := \sup_{q \in [0,1]^d} \left \lvert \frac{D(q,P)}{|P|}   
		- \lambda(q) \right \rvert, 
\label{eq:1}
\end{equation}
where $\lambda(q)$ is the Lebesgue volume of the $d$-dimensional box 
$[0,q)$ and $D(q,P)=|P \cap [0,q)|$ denotes the number of points in $P$ that fall inside this box. 
Thus, intuitively, the $L_\infty$ star discrepancy measures how well the volume of anchored boxes of type $[0,q)$ can be approximated by the fraction of points falling inside this box. 

Other discrepancy notions exist, e.g., differing in the collection $\mathcal{C}$ of subsets $S\subseteq [0,1]^d$ for which the volume shall be approximated (the term ``star'' in the $L_\infty$ star discrepancy indicates $\mathcal{C} = \{ [0,q) \mid q \in [0,1]^d\}$), or in the norm for which the deviation is measured (Definition~\eqref{eq:1} uses the $L_{\infty}$ norm, since we consider the supremum; averaging with respect to the $p$-norm yields another very common class of discrepancy measures, for which in particular the case $p=2$ is well studied). 
Among all discrepancy notions, the $L_\infty$ star discrepancy is the one that has received most attention in the research literature, most notably because of its tight connection to Monte Carlo integration via the Koksma-Hlawka inequality, which states that the absolute additive error of approximating an integral $\int_{[0,1)^d} f(x) \,d\lambda^d(x)$ by the simple average $\frac{1}{|P|} \sum_{p\in P} f(p)$ is bounded from above by $\text{Var}(f) d_{\infty}^*(P)$, where $\text{Var}(f)$ denotes the variation in the sense of Hardy and Krause (see, e.g., \cite{Nie92} for more detailed definitions). In most applications, we cannot control the function $f$ whose integral we aim to approximate, but we typically \emph{do} have control about the points in which we evaluate it. Designing point sets $P$ that minimize $d_{\infty}^*(P)$ is therefore a problem that has been very actively 
investigated
in the last decades.

\paragraph{Known Bounds for the Star Discrepancy} 
Despite significant research efforts spent on analyzing low-discrepancy constructions,  there is still an important gap between the best  known lower and upper bounds. More precisely, it is known that for all $d,n \in \N$ and all sets $P \subseteq [0,1)^d$ of cardinality $|P|=n$, the $L_{\infty}$ star discrepancy of $P$ satisfies 
$d^*_{\infty}(P) \ge \min\{c_0, cd/n \}$, 
where $c_0,c \in (0,1]$ are suitable constants~\cite{Hin04}. On the positive side, there exist $n$-point sets $P$ such that 
$d^*_{\infty}(P) \le  C \sqrt{d/n}$, for some universal constant $C>0$~\cite{HNWW01,GnewuchH20LHSlupper}. 
Uniformly sampled i.i.d.~points satisfy the upper bound in expectation and also with high probability~\cite{HNWW01,Doerr14lowerBoundRandomPoints}. 

In the literature, we often find the following bounds, which focus on the asymptotic dependency on $n$:  
for $d=1$ it holds that every point set $P$ satisfies $d^*_{\infty}(P) \ge 1/(2n)$, and for $d=2$ all $n$-point sets $P$ satisfy 
$d^*_{\infty}(P) \ge 0.023... \log(n)/n$~\cite{SchmidtLowerD2}. It is conjectured that these lower bounds extend to larger dimensions in that for every dimension $d$ there exists a constant $c_d>0$ such that any $n$-point set $P\subseteq [0,1]^d$ satisfies  
$d^*_{\infty}(P) \ge c_d \ln^{d-1}(n)/n$. This conjecture, however, is still open for $d \ge 3$. 
In this case, the best-known lower bound states that for each $d$ there exists a constant $c=c(d)$ such that 
$d^*_{\infty}(P) \ge c \log^{c+(d-1)/2}(n)/n$~\cite{BilykLV08}. 

As can be expected, the point sets that are known to satisfy the best-known upper bounds are  specifically 
tuned to the budget $n$ of points. The discrepancy of a point set $P_n$ which originates from taking the first $n$ elements of a sequence $(P_m)_{m \in \N}$ is necessarily larger, and the above-mentioned conjecture translates to
$d^*_{\infty}(P_n) \ge c_d \ln^{d}(n)/n$, i.e., the best achievable discrepancy is believed to increase by a $\log(n)$ factor. 
Sequences $(P_n)_{n \in \N}$ satisfying 
$d^*_{\infty}(P_n) \le C_d \ln^{d}(n)/n$ for some constant $C_d>0$ are called \emph{low-discrepancy sequences.} 
We will see examples of low-discrepancy sequences in Section~\ref{sec:sequence}. Note, though, that the convergence towards the desired $\ln^{d}(n)/n$ behavior may require very large $n$. It is therefore not clear, a priori, if low-discrepancy sequences are indeed advantageous over random sampling or over other constructions when the sample size $n$ is small. In this case, the first-mentioned type of bounds are more meaningful, but no constructions are known that have a provable advantage over random sampling in all settings $(d,n)$.  
This problem, highly relevant for practical purposes, is, unfortunately, still wide open. 
Even in the two-dimensional case, strictly optimal point sets (i.e., point sets of a given size $n$ which minimize the $L_{\infty}$ star discrepancy) are known only for very small $n \le 6$~\cite{White77}.\footnote{It was pointed out to us by Aicke Hinrichs that some of the constructions derived in this paper are incorrect. The results should thus be handled with care.} 
Finding low-discrepancy point sets for concrete combinations of $d$ and $n$, but without any attempt to find constructions that generalize to other sample sizes or dimensions, has been the focus in~\cite{evolutionary_extension}. 


\textbf{Convention:} Since in this work we will exclusively focus on $L_\infty$ star discrepancies, we shall often drop the explicit mention of the $L_{\infty}$ norm and the explicit mention of the ``star'' property. That is, unless specified otherwise, all occurrences of ``discrepancy'' are to be read as ``$L_{\infty}$ star discrepancy''. 

\subsection{Low-Discrepancy Sequences}
\label{sec:sequence}

We consider five deterministic low-discrepancy sequences 
and two random constructions in our experiments. 
We briefly describe these point sets in the following paragraphs. 

\paragraph{Low-Discrepancy Sequences} 
From the rich set of low-discrepancy sequences (see~\cite{DickP10,Nie92,Mat99} for pointers) we selected the following five, giving preference to constructions which are generally believed to show good behavior in small dimensions. We consider the deterministic variants of each sequence only. Random perturbations (``scrambling'') could yield smaller discrepancy values. However, while the study of such perturbed versions could be interesting in the context of collecting point sets of small discrepancy values, the random nature complicates the interpretation of the results for the subset selection problem, as we shall also see with the two random constructions which we discuss in the next paragraph. 
\begin{itemize}
\item [-] \textbf{Sobol' sequences} ({\tt Sobol}~\cite{Sobol}), also called $(t,d)$-sequences in base 2: For two integers $0<t\le m$, a \emph{$(t,m,d)$-net in base $b$} is a set of points $P=\{p_1,\ldots,p_{b^m}$\} such that for all ``elementary'' boxes $I$ of the form 
$\prod_{j=1}^d [\frac{a_j}{b^{d_j}},\frac{a_j+1}{b^{d_j}})$, 
with $a_j, b \in \N$ satisfying $0<a_j<b^{d_j}$, and volume $\lambda(I)=b^{t-m}$
it holds that $|I \cap P|=b^t$. 
For $t \in \N$, a \emph{$(t,d)$-sequence in base $b$} is a sequence of points $(p_i)_{i \in \N}$ such that for all integers $k>0$ and $m>t$ the set $\{p_{kb^m},\ldots, p_{(k+1)b^m-1}\}$ is a $(t,m,s)$-net in base $b$. 
Various ways to construct Sobol' sequences exist. The most efficient techniques use Gray code representations of integers. Sobol' sequences differ in the initialization numbers, and several works exist, which list good initialization for different dimensions, see~\cite{JK08SobolGneeration} for examples, references, and implementations.
\item [-]
\textbf{Faure sequence} ({\tt Faure}~\cite{Faure82}) is a $(0,d)$-sequence using as prime base the smallest prime number $b$ satisfying $b \ge d$. 
\item [-]
\textbf{Halton sequence} ({\tt Halton}~\cite{Halton64}): Let $b_1,\ldots,b_d > 1$ be co-prime numbers. Define the sequence $P=(p_i)_{i \in \N}$ by setting, for each $j \in [1..d]$, 
$p_i^j:= \sum_{k \ge 0}{d_{j,k}(i)/b_j^{k+1}}$, 
where $(d_{j,k}(i))_{k \in \N}$ is defined as the unique sequence of integers $0 \le d_{j,k}(i) < b_j$ such that 
$i=\sum_{k\ge 0} d_{j,k}(i)b_j^k$. That is, $(d_{j,k}(i))_{k \in \N}$ is the $b_j$-ary representation (also known as \emph{$b_j$-adic expansion}) of $i$, and the Halton points ``inverses'' this representation to obtain numbers in $[0,1]$. 
\item [-] \textbf{Reverse Halton sequence} ({\tt RevHal}): It is known that Halton sequences show some unwanted correlations in the two-dimensional projections (unless the dimension $d$ is very small), see~\cite{DoerrGW14} for an example. To address this shortcoming, different scrambled versions have been suggested. In our experiments we use the {\tt RevHal} constructions suggested in~\cite{RevHal}.
\item [-] \textbf{Fibonacci points} ({\tt Fibon}): This sequence is defined only for the two-dimensional case. 
For a given sample size $n$, the points are defined as $p_i:=\left((\{i/\varphi\}, i/n) \right)_{i \in [1..n]}$, where 
$\{r\}:=r-\lfloor r \rfloor$ denotes the fractional part of the real number $r$ and $\varphi:=(1+\sqrt{5})/2 \approx  1.618$ denotes the golden ratio. 
The Fibonacci sequence is known to satisfy $d^*_{\infty}(P^n) = O(\log b_{n})=O(\log n)$~\cite{Nie92}. Its discrepancy values are hence asymptotically optimal by the already mentioned lower bound for $d=2$ proven in~\cite{SchmidtLowerD2}. 
\end{itemize}

\paragraph{Random Constructions} 
In addition to the five low-discrepancy sequences, we have also considered two randomized constructions, uniform sampling and Latin Hypercubes. While the former can be seen as a sequence, the latter requires to fix the number of points $n$ in advance, so that it is not referred to as a sequence, but a point set. 
\begin{itemize}
\item [-] \textbf{Uniform sampling} ({\tt unif}): We simply select $p_i \in [0,1]^d$ uniformly at random, and do this independently for each $i$.  

\item [-] \textbf{Improved Latin Hypercube Sampling} ({\tt iLHS}): Classical Latin Hypercube sampling requires to sample $d$ permutations $\sigma^1, \ldots, \sigma^d$ of the set $[1..n]$   
and to set $p^{j}_i:=(\sigma^j(i)-u^{j}_i)/n$, where  $0\le u^{j}_i <1$ denotes a uniformly sampled value. That is, we select the $i$-th point $p_i$ by choosing it randomly in the box $[(\sigma^j(i)-1)/n, \sigma^j(i)/n]^d$. The advantage of LHS over uniformly selected (Monte Carlo) points is that the one-dimensional projections are all well spread. A disadvantage is that the points can nevertheless be close to each other, e.g., when $\sigma^j$ is the identity permutation for all $j \in [1..d]$ (in which case the $n$ points are all close to the diagonal). Various versions of LHS have been suggested in the literature. In our comparison we use the ``improved'' LHS construction suggested in~\cite{iLHS}. This variant constructs the  set $P$ iteratively, by sampling at each stage a few alternatives and then selecting the candidate that maximizes the distance to the points that are already collected in the set $P$. 
\end{itemize}

In terms of discrepancy values the two randomized constructions behave quite similarly: as already commented in Section~\ref{sec:discrepancy-def}, the expected star discrepancy of a set $P$ of $n \ge d$ i.i.d.~uniformly selected points is bounded from below by
$
{\bf {\rm E}} [d^*_{\infty}(P)] \ge c \sqrt{d/n}
$, for some universal constant $c>0$~\cite{Doerr14lowerBoundRandomPoints}. It is furthermore unlikely that the star discrepancy of $P$ is much smaller, as the concentration bound 
$
{\bf {\rm P}}\! \left(d^*_{\infty}(P)< c \sqrt{d/n} \right) \le \exp(-\Omega(d)),
$ also proven in~\cite{Doerr14lowerBoundRandomPoints}, shows. 
This bound is tight, in the sense that there exists a constant $C$ such that the star discrepancy of a uniform i.i.d.~point set $P$ satisfies 
$d^*_{\infty}(P)] \le C \sqrt{d/n}$ with high probability, see~\cite{randomUpperbounds} for an explicit proof and further references.  
The same bounds also apply to LHS with randomly placed points in the selected boxes, provided that $d \ge 2$ and $n \ge 1\,600d$~\cite{DoerrDG18LHSbounds,GnewuchH20LHSlupper}. Note that in our experiments, we deal with much smaller sample sizes and we use the ``improved LHS'' suggested in~\cite{iLHS}, for which the results do not immediately apply. 

\subsection{Computation of Star Discrepancy Values}
\label{sec:computation}

We finally summarize a few computational aspects of star discrepancy computation.  
A more exhaustive survey on this topic can be found in the book chapter~\cite{DoerrGW14}. 
The decision version of calculating the star discrepancy of a point set is
an NP-hard problem~\cite{complexity} as well as W[1]-hard~\cite{w1_discrepancy}.
The most efficient algorithm for this problem was proposed in \cite{discrepancy_algorithm},
with a running time of \(O(n^{d/2 + 1})\). For $d \in \{2,3\}$ the algorithm proposed in~\cite{BZ93} is efficient but its running time scales as $n^d/d!$ when extending the algorithm to higher dimensions.

At the heart of all algorithms designed to evaluate the star discrepancy of a given point set $P=\{p_1,\ldots,p_n\} \subseteq [0,1]^d$ is the following observation, which reduces the maximization problem~\eqref{eq:1} to a discrete problem. 
Our subset selection algorithms will make use of these observations, and, therefore, we briefly summarize this reduction. 

For any point $q \in
[0,1]^d$ let $D(q,P)$ be the number of points that fall inside the open box $[0,q)$ and let $\overline{D}(q,P)$ be the number of points inside the closed box $[0,q]$, respectively. That is,  
\begin{equation}
	D(q,P) := \sum_{i = 1}^{n} \mathbf 1_{[0, q)} (p_{i})   
\qquad \text{ and } 
\qquad 
\overline{D}(q,P) := \sum_{i = 1}^{n} \mathbf 1_{[0, q]} (p_{i}), 
\end{equation}
\noindent where $\mathbf 1$ denotes the indicator function, i.e., $\mathbf 1_{[0, q)} (p)=1$ if $p\in [0,q)$ and $\mathbf 1_{[0, q)} (p)=0$ otherwise. 
The \emph{local discrepancy} $d_{\infty}^{*}(q,P)$ of a point $q \in [0,1]^d$ is defined as the maximum of the following two values:
\begin{equation}
	\delta(q, P) := \lambda(q) - \frac{1}{n} D(q, P)
\qquad \text{ and } 
\qquad 
	\overline{\delta}(q, P) := \frac{1}{n} \overline{D}(q, P) - \lambda(q).
	\label{eq:delta}
\end{equation}
Instead of evaluating $\delta(q, P)$ and $\overline{\delta}(q, P)$ for all $q \in [0,1]^d$, it suffices to consider the points on the following grids: 
\begin{equation}
\Gamma(P) := \Gamma^{1}(P) \times ... \times \Gamma^{d}(P) 
\qquad \text{ and } 
\qquad 
\overline{\Gamma}(P) := \overline{\Gamma}^{1}(P) \times ... \times \overline{\Gamma}^{d}(P),
\end{equation}
\noindent where, for $j\in[1..d]$,
\begin{equation}
\Gamma^{j}(P):=\{p^{j}_{i} \ \mid \ i \in \{1, \ldots, n\}\}
\qquad \text{ and } 
\qquad 
\overline{\Gamma}^{j}(P) :=\Gamma^{j} \left(P \cup \{(1,\ldots, 1)\} \right).
\end{equation}
The reduction of Equation~\eqref{eq:1}  to a discrete maximization problem then states that 
\begin{equation}
d_{\infty}^{*}(P) := \max \left\{\max_{q \in \overline{\Gamma}(P)}\delta(q, P),
\max_{q \in \Gamma(P)}\overline{\delta}(q, P)\right\}.
\label{discrepancy_formula}
\end{equation}
A simple proof for this equation can be found in~\cite{complexity}, but, as mentioned there,  
Equation~\eqref{discrepancy_formula} had previously been mentioned in several works, with~\cite{Nie72a} being one of the earliest examples.

\section{The Star Discrepancy Subset Selection Problem}
\label{sec:subset}

We are now ready to define the \emph{star discrepancy subset selection problem}. 
Given a $d$-dimensional point set $P \subseteq [0,1]^d$ of size $\lvert P \rvert = n$,
and given an integer \(m \le n\), the goal is to find a subset \(P^{*} \subseteq P\)
of size \(\lvert P^*\rvert = m\) such that \(d_{\infty}^{*}(P^{*})\) is
minimized. 
Using Equation~\eqref{discrepancy_formula}, the $L_{\infty}$ star discrepancy subset selection problem has the following equivalent formulation: 
\begin{equation}
	\label{eq:10}
\min_{\substack{P^* \subseteq P\\|P^*| = m}} \max \left \{\max_{q \in \overline{\Gamma}(P^{*})}\delta(q, P^*), \max_{q \in \Gamma(P^{*})}\overline{\delta}(q, P^*) \right\}.
\end{equation}
We show NP-hardness of this problem in Section~\ref{sec:NP} and we discuss some basic properties in Section~\ref{sec:Basic prop}.

\subsection{NP-Hardness of the Subset Selection Problem}
\label{sec:NP}
We consider the following decision version of the $L_{\infty}$ star discrepancy subset selection problem:

\fbox{\parbox{\textwidth}{%
\textbf{Decision problem:} Discrepancy Subset Selection \\
\textit{Instance}: natural numbers $n, m \in \mathbb{N}$, $m \leq n$, $\varepsilon \in \left(0,1 \right]$, point set $P=(p^i)_{i \in [1..n] }$ \\
\textit{Question}: Is there $P' \subseteq P$ such that $|P'|=m$ and $d_{\infty}^{*}(P') \leq \varepsilon$?
}}

\begin{theorem}\label{thm:NP}
The decision version of the discrepancy subset selection problem is NP-hard.
\end{theorem}

Similarly as in the proof of the NP-hardness of calculating the $L_{\infty}$ star discrepancy presented in~\cite{complexity}, we will obtain Theorem~\ref{thm:NP} by a reduction from the DOMINATING-SET problem, a problem that is well known to be NP-complete (see, for example, \cite{Domi}).

\fbox{\parbox{\textwidth}{%
\textbf{Decision problem:} DOMINATING-SET \\
\textit{Instance}: Graph $G=(V,E)$, $m \in [1..|V|]$\\
\textit{Question} Is there a set $T \subseteq V$ of size at most $m$ such that for any $v \in V \setminus T$, there exists $t \in T$ such that $(v,t) \in E$?
}}

\begin{proof}[Proof of Theorem~\ref{thm:NP}]
Throughout this proof, let $q \in \left[ 0, 1 \right]^n$. We consider an instance $G=(V,E)$, $m \in [1..|V|]$ of DOMINATING-SET. We build a point set $P=(p_{i})_{i \in [1..n]}$ in $\mathbb{R}^{n}$ where $n=|V|$ by defining, for all $i,j \in [1..n]$, 
$$ p_i^j:=\begin{cases} \alpha ~~ & \text{if}~ (i,j) \in E ~~\text{or}~~i=j,\\
    \beta ~~ & \text{otherwise},\end{cases}  $$
    where $\alpha$ and $ \beta$ are two real values such that $\frac{1}{n}>\alpha>\beta>0$. 
    
To begin, we introduce the following formula for the $L_{\infty}$ star discrepancy, shown in \cite{complexity}. For any point set $P$, 
\begin{equation}\label{Vk}
    d_{\infty}^*(P)=\max \big\lbrace \max_{k=0,...,n-1} (V^k_{\max}- \tfrac{k}{n}), \max_{k=1,...,n} (\tfrac{k}{n}-V^k_{\min}) \big\rbrace,  
\end{equation}
where $V^k_{\min}$ is the volume of the smallest (by the Lebesgue measure) closed box containing at least $k$ elements of $P$, and $V^k_{\max}$ is the volume of the largest half-open box containing at most $k$ elements of $P$.

We will set aside the case $k=n$ in the first part of the proof. We consider $P_T$, a subset of $P$ of size $m$ associated to a subset $T \subseteq V$ of size $m$. For this point set $P_T$, the largest empty box $V^0_{\max}$ is of size at least $\beta^n$ and at most $\alpha^n$. Since the maximal coordinate for any point in $P_T$ is $\alpha$, any half-open box that does not contain all the points of $P_T$ will have at least one coordinate smaller than $\alpha$. This gives us $V^k_{\max} \leq \alpha$ for all $k \in [1..n-1]$. We obtain the upper-bound $\alpha$ for the first maximum in (\ref{Vk}) by choice of $\alpha$ and $\beta$.

For the second maximum, any closed box containing at least one point (but not all of them) will have each coordinate greater than or equal to $\beta$ since the lowest coordinate in any dimension for each point is $\beta$. This gives us a minimum volume of $\beta^n$ for a box containing some points of $P_T$. The fraction $\tfrac{k}{n}$ can be upper-bounded by $\tfrac{n-1}{n}$, which gives us an upper bound of $\tfrac{n-1}{n}- \beta^n \geq  \max_{k=1,...,n-1} (\tfrac{k}{n}-V^k_{\min})$.

We now consider the case $k=n$ depending on whether or not $T$ is a dominating set of $G$. If $T$ is a dominating set, by our point set construction, for any $j \in [1..n]$, there exists $i \in T$ such that $p_j^i=\alpha$. Any box of the type $[0,q]=\prod_{i=1}^n [0,q_i]$ will contain all the points of $P_T$ if and only if for all $i \in [1..n], q_i \geq \alpha$. Therefore, the smallest box containing all the elements of $P_T$ has volume $\alpha^n$. This gives us a $1-\alpha^n$ term in (\ref{Vk}), which is greater than all the other terms calculated previously, by choice of $\alpha$ and $\beta$.

If $T$ is not a dominating set, there exists a vertex $i$ not dominated by the elements of $T$. Since $i$ is not dominated by $T$, the smallest full-box has size at most $\alpha^{n-1} \beta$ since all the points in $P_T$ have $\beta$ as their $i$-th coordinate. This gives us at least a $1-\alpha^{n-1} \beta$ term in Equation (\ref{Vk}) which like in the previous case is also greater than all the other terms. It is also strictly greater than $1-\alpha^n$. In both cases, we have shown that the discrepancy value is obtained by the volume of the smallest full-box. We have that $d_{\infty}^*(P_T) \leq 1- \alpha^n$ if and only if $T$ is a dominating set of $G$. We note that any dominating set $T$ of size strictly smaller than $m$ can become a dominating set of size exactly $m$ by adding points from $G$ until $T$ is of size exactly $m$. This gives us the desired result: $G$ has a dominating set of size at most $m$ if and only if $P$ has a subset of size $m$ of discrepancy at most $1- \alpha^n$.
\end{proof}

We note that while the problem is NP-hard, it is not NP-complete to our knowledge. If we are given a subset of a point set $P$, checking if its discrepancy is smaller than $\epsilon$ cannot be done in polynomial time to our knowledge, under the hypothesis that $\mathtt{P} \neq \mathtt{NP}$. This comes from the fact that we want an upper bound on the discrepancy and not a lower bound. For a lower bound, given a specific anchored box, we can verify that the discrepancy is large enough in linear time by counting the points in the box and calculating its volume (Star discrepancy was shown to be NP-complete in \cite{complexity}). On the other hand, for an upper bound, we would need to check that all the possible anchored boxes have a small enough local discrepancy. It is not sufficient to exhibit one of them.

\subsection{Other Basic Properties of the Discrepancy Subset Selection Problem}
\label{sec:Basic prop}

\paragraph{Non-Monotonic Behavior of the Star Discrepancy} 
Before we discuss our strategies to solve the star discrepancy subset selection problem, we first note that the star
discrepancy is a non-monotone function, in the sense that $P' \subseteq P$ does not imply any order of $d_{\infty}^{*}(P')$ and $d_{\infty}^{*}(P)$.  
The following example illustrates this non-monotonic behavior. It is visualized in Figure~\ref{fig:example}. 
\begin{example}\label{ex}
	Let $P := \{(0.1,0.4),(0.2,0.9),(0.7,0.6),(0.8,0.7)\}$.
	Then, it holds that $d_\infty^* \left(P \right)= 0.40$, whereas 
	$d_\infty^* \left(P \cup \{(0.9, 0.2)\} \right) = 0.43$ and
	$d_\infty^* \left(P \cup \{(0.3, 0.3)\} \right)= 0.33$. 
	As we can see in Figure~\ref{fig:example}, the discrepancy value of the first set is determined by the point $q=(1.0,0.4)$, whereas it is determined by points $(0.7,0.9)$ and $(0.3,0.9)$ in the second and third case, respectively. 
\end{example}

\begin{figure}
    \centering
    \includegraphics[width=0.325\textwidth]{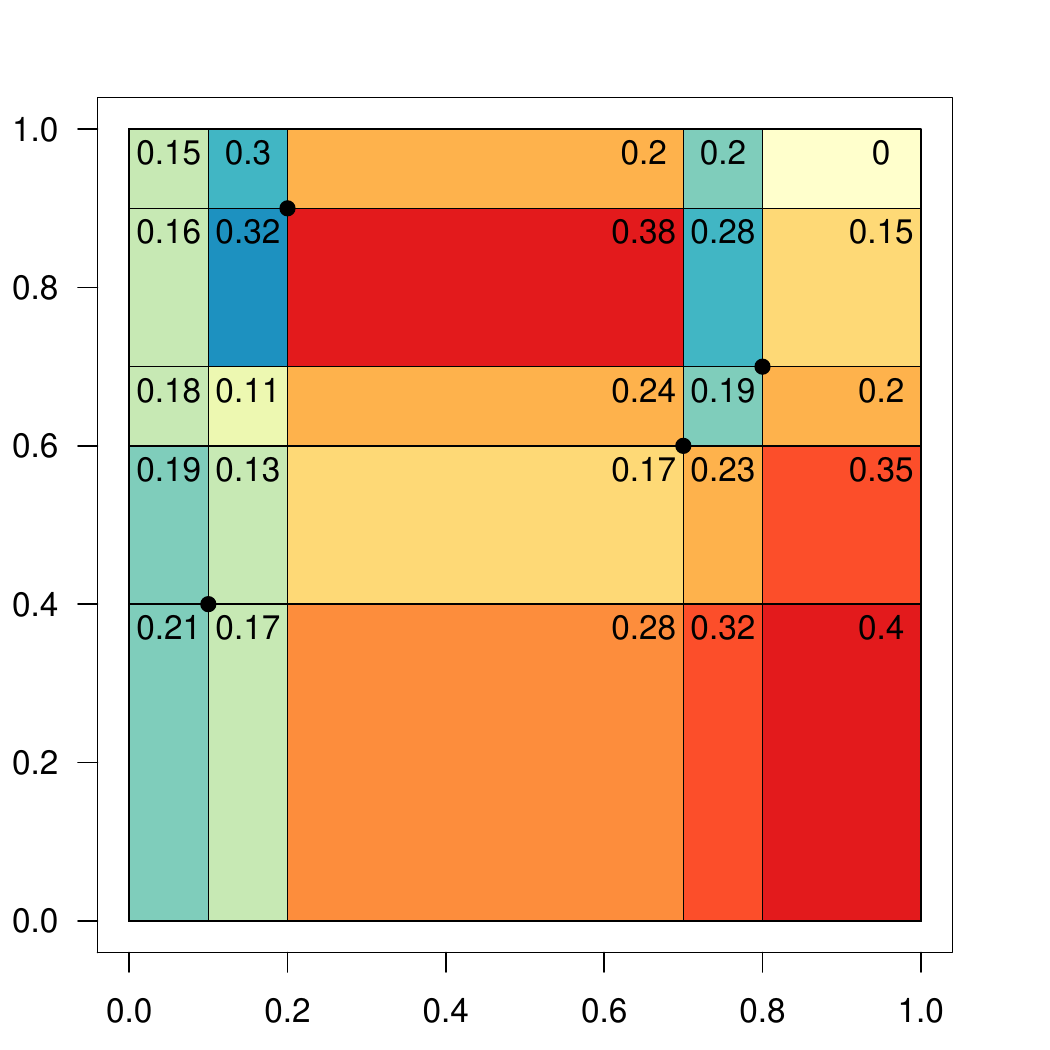}
    \includegraphics[width=0.325\textwidth]{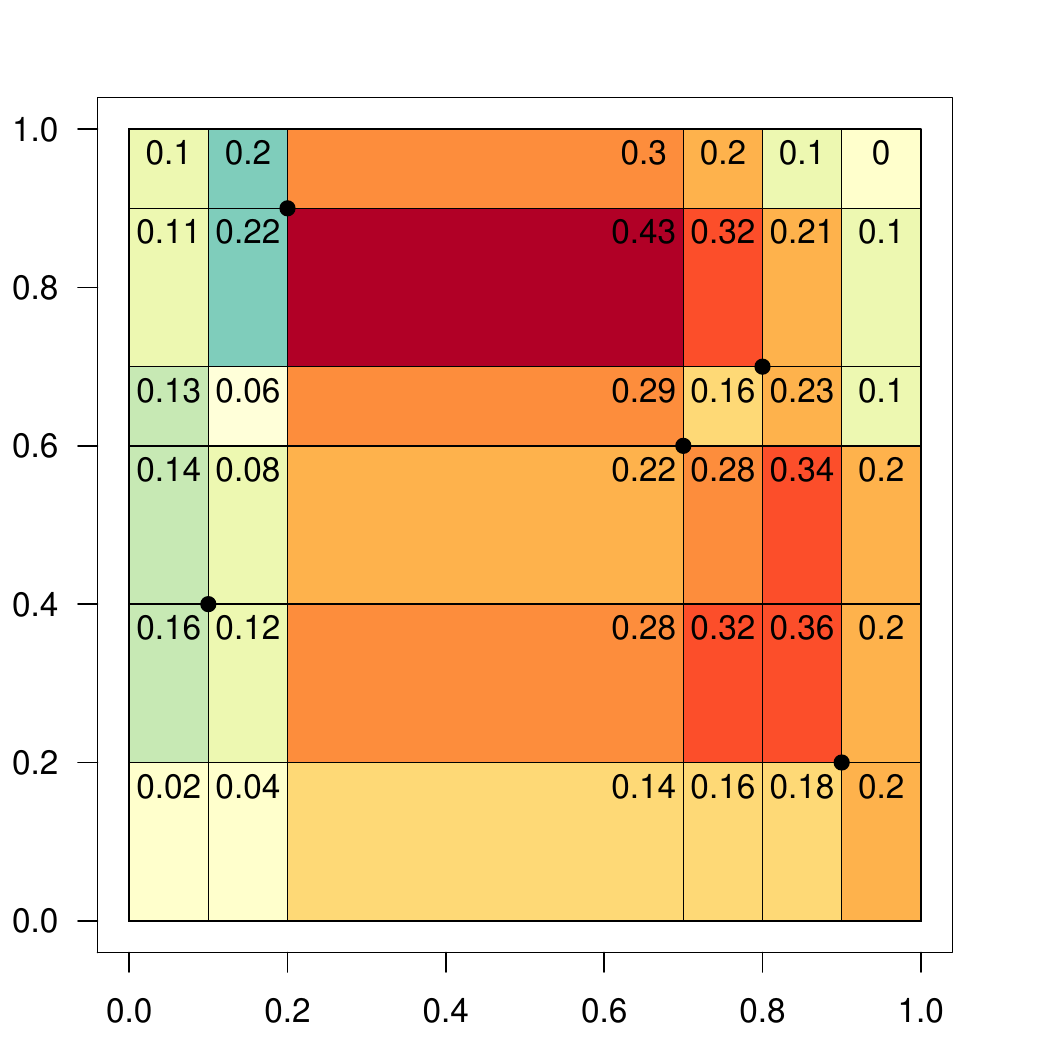}
    \includegraphics[width=0.325\textwidth]{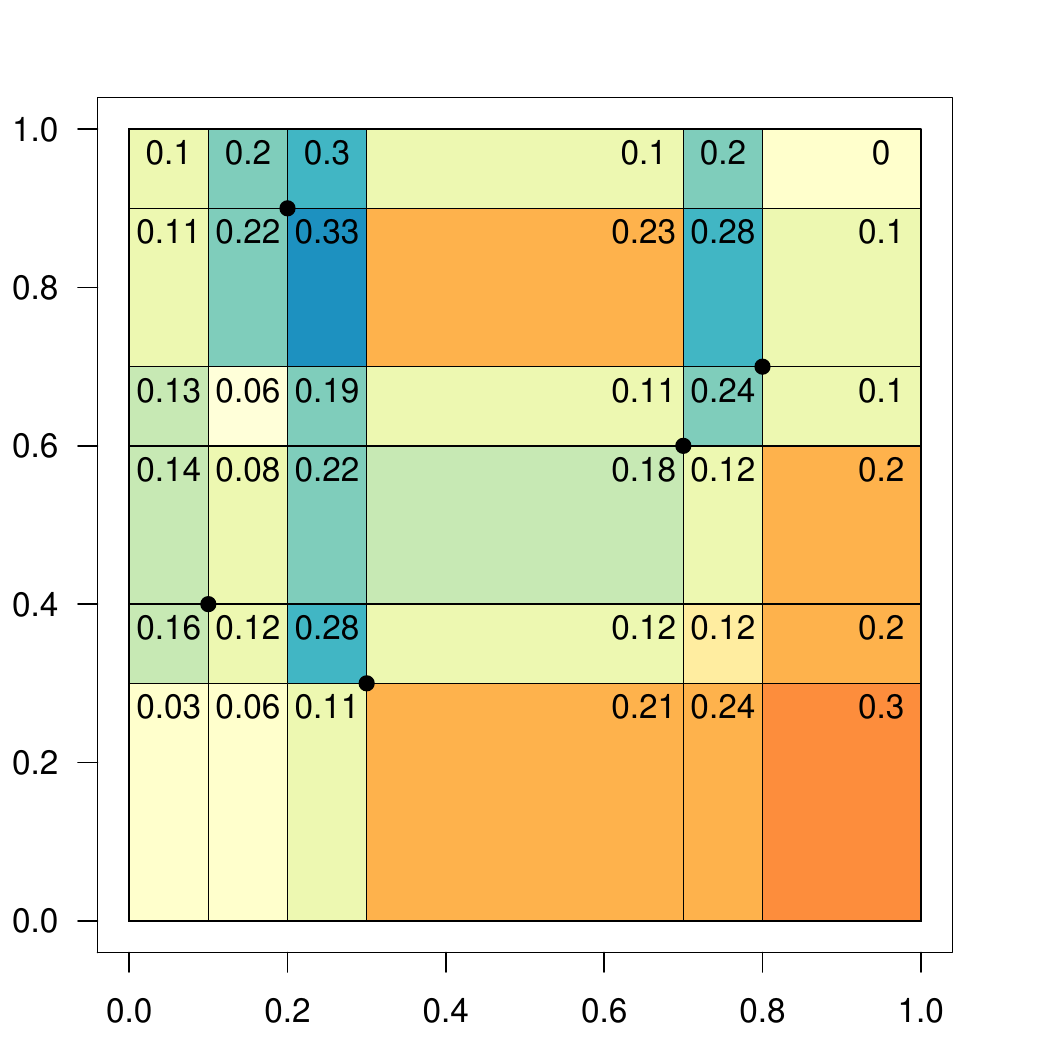}
    \caption{Illustration of the point sets defined in Example~\ref{ex}. The values are the local star discrepancy values and the color of each cell is the local star discrepancy value of its upper right corner (darker colors are used for larger discrepancy values). Blue colors are used when $d^*_{\infty}(y,P)=\overline{\delta}(y,P)$, and red colors are  used when $d^*_{\infty}(y,P)=\delta(y,P)$. 
    }
    \label{fig:example}
\end{figure}

\paragraph{Non-Monotonic Behavior of the Subset Selection Problem}
For a fixed set of points $P$ (therefore, fixed $n$), there is also no relation between the smallest discrepancy subsets obtained for different $m$. To illustrate this non-monotonic behavior, consider the following simple example in dimension~1. Let $P:=\{\tfrac{1}{6},\tfrac{1}{4},\tfrac{1}{2},\tfrac{3}{4},\tfrac{5}{6}\}$. Niederreiter introduced in~\cite{Nie72a} an explicit formula for the discrepancy in dimension~1, also showing that for an $n$-point set, the minimal discrepancy is uniquely obtained by the set $\{\tfrac{1}{2n},\tfrac{3}{2n},...,\tfrac{2n-1}{2n}\}$. The optimal subset of size 2 of $P$ is therefore given by $P_2'=\{\tfrac{1}{4},\tfrac{3}{4}\}$, whereas the optimal subset of size 3 is given by $P_3'=\{\tfrac{1}{6},\tfrac{1}{2},\tfrac{5}{6}\}$.

Another example for this non-monotonic relationship between optimal subsets of different sizes can be found in the results presented in this paper. For example, in the 2-dimensional setting, the best size $m=20$ subset of the first $n=140$ elements of the Fibonacci sequence is not contained in the best subset of size $m=120$.\footnote{We do not discuss the position of the points in these optimal subsets here in this work, but the sets can be found in the repository available at~\url{https://algo.dei.uc.pt/star}.}

\paragraph{Extending the Grid to its Original Size} 
For the design of our mixed-integer linear formulation, it will be convenient to consider the whole grid induced by $P$, and not only the one induced by $P^*$. 
\begin{lemma}
\label{lem:grid}
The star discrepancy subset selection problem is equivalent to the following:
\begin{equation}
	\label{eq:2}
\min_{\substack{P^* \subseteq P\\|P^*| = m}} \max \left \{\max_{q \in \overline{\Gamma}(P)}\delta(q, P^*), \max_{q \in \Gamma(P)}\overline{\delta}(q, P^*) \right\}.
\end{equation}
\end{lemma}

\begin{proof}
Let $P^* \subseteq P$, $|P^*| = m$. 
We show that  
$\max_{q \in \overline{\Gamma}(P)}\delta(q, P^*) = \max_{q \in \overline{\Gamma}(P^*)}\delta(q, P^*)$ and that 
$\max_{q \in \Gamma(P)}\overline{\delta}(q, P^*) = \max_{q \in \Gamma(P^*)}\overline{\delta}(q, P^*)$. Since $P^* \subseteq P$, we only need to prove ``$\le$''. 
To show the first equation, let $q \in \overline{\Gamma}(P)$. 
For every coordinate $j \in [1..d]$ let $u^j:=\min\{u \in \overline{\Gamma}^j(P^*) \mid u \ge q^j\}$. 
Then $D(u,P^*)=D(q,P^*)$ and hence
$\delta(u,P^*) 
= \lambda(u) - \frac{1}{|P^*|}D(u,P^*) 
\ge \lambda(q) - \frac{1}{|P^*|}D(q,P^*)
= \delta(q,P^*)$. 
This shows that $\max_{q \in \overline{\Gamma}(P)}\delta(q, P^*) \le \max_{q \in \overline{\Gamma}(P^*)}\delta(q, P^*)$. 

For the second equation, let $q \in \Gamma(P)$. 
Set $\ell^j:=\max\{\ell \in {\Gamma}^j(P^*) \cup\{0\} \mid \ell \le q^j\}$. 
Then $\overline{D}(\ell,P^*)=\overline{D}(q,P^*)$ and thus 
$\overline{\delta}(\ell,P^*) 
= \frac{1}{|P^*|}\overline{D}(\ell,P^*) - \lambda(\ell) 
\ge \frac{1}{|P^*|}\overline{D}(q,P^*) - \lambda(q) 
= \overline{\delta}(q,P^*)$.
\end{proof}

\section{Algorithmic Approaches to solve the Discrepancy Subset Selection Problem}
\label{sec:algos}

In this section, we suggest two different approaches to solve the star discrepancy subset selection problem, one based on mixed-integer linear programming (Section~\ref{sec:milp}) and one based on branch and bound (Section~\ref{sec:BBalgo}). This section introduces these exact solvers. In Section~\ref{sec:comp-algos}, we 
compare their performance against two heuristics, random subset selection and an iterative greedy selection, which we use to obtain an initial solution for the branch and bound algorithm. The greedy approach is described in Section~\ref{sec:greedy}.

\textbf{Convention:} To ease the description of our algorithms, we assume that, for all $j \in [1..d]$, the coordinates $\{p_i^j \mid i \in [1..n]\}$ are pairwise different.

\subsection{A mixed integer linear programming formulation}
\label{sec:milp}

For simplification
purpose, we start with the description of the mixed integer linear programming 
(MILP) model for the star discrepancy subset selection problem for $d=2$. 
We then discuss extensions for larger 
dimensions. 

The following component-wise order relations in \(\mathbb R^d\) will be
required for our model. For $v, w \in \mathbb R^d$, we write 
\begin{align*} 
        v \leqq w & \iff v^j \leq w^j \;\text{ for all } j\in[1..d]\\
	v \leq  w & \iff v\neq w \mbox{ and } v \leqq w \\
	v <     w  & \iff v^j < w^j    \;\text{ for all }  j\in [1..d] 
\end{align*} 

Consider a two-dimensional point set 
$P := \{p_1, p_2,\ldots, p_{n}\} \subseteq \mathbb R^2$.
Without loss of generality, we assume the points in $P$ are 
reordered such that $p_1^1 \leq p_2^1 \leq \cdots \leq p_n^1$. 
Let $\mathcal S_n$ denote the symmetric group of order $n$ and $\sigma \in \mathcal S_n$
denote a permutation of $[1..n]$, such that $p_{\sigma(1)}^2 \leq p_{\sigma(2)}^2 \leq \cdots \leq p_{\sigma(n)}^2$.\footnote{Where our convention of pairwise different coordinates does not apply, we assume the following: In the case of a tie $p_i=p_{i+1}$ in the first coordinate, we assume that $p_i^2 \leq p^2_{i+1}$. In the case of a tie in the second coordinate $p_{\sigma(i)}^2=p_{\sigma(i+1)}^2$ we assume that that $p^1_{\sigma(i)} \leq p^1_{\sigma(i+1)}$.}

We shall use $\gamma_{i,j}(P)$ to denote the grid point at position $(i,j)$ in
$\overline \Gamma(P)$, $i,j \in [1..n+1]$.\footnote{For simplicity, we restrict the 
presentation to $\overline \Gamma(P)$, since $\Gamma(P) \subset \overline \Gamma(P)$.}
Then, due to the ordering 
of the points in $P$, it holds that 
$\gamma_{i,\sigma(i)}(P) = p_i$, 
$\gamma_{i,n+1}(P) = (p_i^1,1)$,
$\gamma_{n+1,\sigma(i)}(P) = (1, p_{\sigma(i)}^2)$,
for $i\in[1..n]$, and $\gamma_{n+1,n+1}(P) = (1,1)$.
In addition, we define the following index sets 
$$\overline \Delta(P,i,j) := \{ \ell \mid p_\ell \leqq \gamma_{i,j}(P), p_\ell \in P \}
\qquad \text{ and } \qquad
\Delta(P,i,j) := \{ \ell \mid p_\ell < \gamma_{i,j}(P), p_\ell \in P \}$$
for $i,j\in[1..n+1]$. 

For our MILP model, we define a binary variable $x_{i}$ 
that takes value $1$ if point $p_i$ is selected,
$i \in [1..n]$, and $0$ otherwise. 
The model is as follows.

\begin{equation}\label{eq:ilp2}
    \begin{array}{>{\displaystyle}r>{\displaystyle}r@{\extracolsep{0.75ex}}>{\displaystyle}c@{\extracolsep{0.75ex}}>{\displaystyle}l@{\extracolsep{2em}}>{\displaystyle}l}
\min  &  z\\
      \text{s.\,t.} 
      & z & \geq & h_{i,j}-\frac{1}{m} \sum_{\ell \in \Delta(P, i, j)} x_{\ell} & \;\text{ for all } i,j \in [1..n+1] \\
      & z & \geq & - h_{i,j}+\frac{1}{m} \sum_{\ell \in \overline \Delta(P, i, j)} x_{\ell} & \;\text{ for all }  i,j \in [1..n] \\
      & \sum_{i=1}^n x_{i} & = & m \\	 	   
      & x_{i} &\in&\{0,1\} & \;\text{ for all }
      i \in [1..n]\\
      & z & \in & \mathbb R_{\geq 0}\\
	   \end{array}
  \end{equation}

Variable $z$ is a non-negative continuous variable that takes the optimal 
star discrepancy value. 
The first two constraints are
due to the linearization of the objective function in Problem \eqref{eq:2} and
bound the minimum value of $z$, where 
$h_{i,j}$ is the measure of the $d$-dimensional box $[0,\gamma_{i,j}(P)]$. 
The third constraint ensures that exactly
$m$ points in $P$ are selected.

  Extending our MILP for more dimensions, we obtain 
$O(n^d)$ constraints and $O(n)$ variables. Noteworthy, its relaxation,
that is, $x_i \in [0,1]$,
has an integral solution when $m=n$. This suggests that
the \emph{integrality gap}, i.e., the difference between the optimal value $z^*$ for the original MILP (satisfying that there exists $x^* \in \{0,1\}^n$ such that the conditions are satisfied) and the optimal value $z^*_{\text{relax.}}$ of its relaxation (where we only require existence of $x \in [0,1]^n$ for which the conditions are satisfied),  
may be small when the ratio $m/n$ is large. This suggestion is confirmed in the experimental results reported in Section \ref{sec:comp-algos}.

\subsection{A combinatorial branch-and-bound algorithm}
\label{sec:BBalgo}

Algorithm \ref{pseudocode} presents the pseudocode of our combinatorial
branch-and-bound approach (BB) for the star discrepancy subset selection problem,
for a given $m$ and a given point set $P$. 
At a given iteration, the algorithm maintains three stacks: $S_A$, which
stores the points that were accepted (subset $P_A$), $S_B$, which 
stores the points that were rejected (subset $P_R$), and $S_N$, which 
stores the points for which a decision has not yet been taken (subset $P_N$). 
Both $S_A$ and $S_R$ are
empty in the beginning, whereas $S_N$ contains all points in $P$.  Variable
$ub$ corresponds to the lowest upper bound on the optimal discrepancy value
found so far and is initially set to 1, which clearly is an upper bound for the star discrepancy of any subset of $P$, since it is an upper bound for the star discrepancy of \emph{any} point set.

The branching part of the algorithm works as follows: at each recursive step, the point $p$ at the top of stack $S_N$ is removed and is 
placed at the top of stack $S_A$ ($p$ is accepted). Then, the 
remaining smaller sub-problems are solved recursively, with the points 
in $S_A$ (and none of the points of $S_R$) belonging to the solutions of these sub-problems.
When back to the same recursion level, the point $p$ is removed from 
the top of stack $S_A$ and placed at the top of stack $S_R$ ($p$ is
rejected) and the same procedure is repeated again for the smaller
sub-problems.

The usual stopping conditions avoid the generation of infeasible solutions, namely, either having $m$ points or there are not enough points in $S_N$ to reach a solution with $m$ points. In the former case, the star discrepancy value of the $m$ points is computed and compared
against the upper bound ($ub$), which is updated accordingly. The function $LB(\cdot)$ allows the
pruning of the search tree by computing a lower bound on the smallest value of
star discrepancy of a feasible solution that contains the points stored in
$S_A$.  The following section describes the lower bound computations.

\begin{algorithm}[t]
\caption{Branch and Bound}
\label{pseudocode}
\begin{algorithmic}\STATE \(S_A := \emptyset\), \(S_R := \emptyset\), \(S_N := (p_1, ..., p_n)\), $ub: = 1$ 
\end{algorithmic}
{\bf Function} $BB(S_A, S_R, S_N)$
\begin{algorithmic}[1]
\IF {$\lvert P_A \rvert = m$}
\STATE $ub := \min \{ub, d_\infty^*(P_A)\}$
    \RETURN
\ELSIF {$P_N = \emptyset \mbox{ or } \lvert P_A \rvert + \lvert P_N \rvert < m$}
    \RETURN
\ELSIF {$LB(P_A, P_R, P_N) > ub$}
    \RETURN
\ELSE
\STATE $q := \mbox{pop}(S_N)$
\STATE $\mbox{push}(q,S_A)$
\STATE $BB(S_A, S_R, S_N)$
\STATE $p := \mbox{pop}(S_A)$
\STATE $\mbox{push}(p,S_R)$
\STATE $BB(S_A, S_R, S_N)$
\STATE $q := \mbox{pop}(S_R)$
\STATE $\mbox{push}(q,S_N)$
\RETURN
\ENDIF
\end{algorithmic}
\end{algorithm}

\subsubsection{Lower bounds}
Consider that, at a given moment of Algorithm 1, 
stacks $S_A$, $S_R$, and $S_N$ contain point sets $P_A$, $P_R$, and $P_N$, 
respectively. Note that $P_N := P \setminus (P_A \cup P_R)$. 
Let $P_A^*$ be the set of $m$ points with the smallest
value of star discrepancy that contains $P_A$ and does not intersect $P_R$, that
is, 
\begin{equation}
P_A^* := \arg\min \left \{ d^*_\infty(P^\prime) \  |  \ P_A \subseteq P^\prime \subseteq P \setminus P_R, |P^\prime| = m \right\}. 
\end{equation}

Our bounding function is the maximum of two values, that is, 
$$LB(P_A,P_R,P_N) = \max \{ LB_1(P_A,P_R,P_N),LB_2(P_A,P_R, P_N)\}.$$

The second value, $LB_2(P_A,P_R, P_N)$, is a lower bound on the 
local discrepancy of points in $\Gamma(P_A)$, \begin{equation}
	LB_2(P_A,P_R, P_N) := \max_{q\in \Gamma(P_A)} \left\{ \frac{1}{m} \overline{D}(q, P_A) - \lambda(q)\right\} \leq d^*_\infty(P^*_A).
\end{equation}
Note that $\overline{D}(q,P_A) \leq \overline{D}(q,P_A^*)$ holds
for every point $q \in \Gamma(P_A^*)$.

The first value, $LB_1(P_A, P_R, P_N)$, is also a lower bound on the local
discrepancy  of points in $\overline \Gamma(P_A)$.  
For a given set $P_A$ and set $P_N$, at each point $q$ in $\overline
\Gamma(P_A)$, an upper bound on the value of $D(q,P_A^*)$ is as follows
\begin{equation}
\label{eq:lb1}
	\min \left\{m, D(q,P_A) + D(q,P_N) \right\} \geq D(q, P_A^*)
\end{equation}
from which the following lower bound $\eta(q, P_A, P_N)$ on the value of the local discrepancy at 
point $q$ can be derived:
\begin{equation}
	\eta(q, P_A, P_N) := \lambda(q) - \frac{1}{m} \min\left\{m, D(q,P_A)+D(q,P_N)\right\} \leq \delta(q,P_A^*). 
\end{equation}
Finally, we define our lower bound $LB_1(P_A,P_R, P_N)$ as follows
\begin{equation}
	LB_1(P_A, P_R, P_N) := \max_{q \in \overline \Gamma(P_A)} \left\{\eta(q, P_A, P_N)\right\}\ 
	\leq  d^*_\infty(P_A^*).
\end{equation}

\subsubsection{Lower bound computation}

A na\"ive computation of $LB(P_A,P_R,P_N)$ requires 
$O(dm^{d+1})$ time
for each new point $p$ that enters into $P_A$ 
due to the need to compute 
$D(q,P_A)$ and $D(q,P_N)$ 
for every point $q$ in $\overline \Gamma(P_A)$ as well as  
$\overline{D}(q^\prime,P_A)$ and 
$\overline{D}(q^\prime,P_N)$
for every point $q^\prime$ in $\Gamma(P_A)$.
However, note that $D(q, P_N)$ and
$\overline D(q^\prime, P_N)$ can be pre-computed 
for every point $q \in
\overline \Gamma(P)$ and every point $q^\prime \in\Gamma(P)$,
respectively, and for every subset $P_N$, which 
requires $\Theta(n^{d+1})$ space and time in the pre-processing step.
Moreover, for each new point $p$ that enters into $P_A$ at each
recursive step, $D(q, P_A)$ and $\overline D(q^\prime, P_A)$ need only to be
computed at points $q \in \overline \Gamma(P_A)$ and $q^\prime \in \Gamma(P_A)$
such that $p \leq q$ and $p \leq q^\prime$ holds, respectively.
In the following, we show that both components of the lower bound 
$LB(P_A,P_R,P_N)$ can be computed even faster in practice by considering a
given ordering of the points in $P$.

As done in Section~\ref{sec:milp} and assuming that the points have pairwise different coordinates (as per our convention made at the beginning of this section), let us assume, without loss of generality, that $p_i^1  < p_j^1$ holds for $i < j$, $i\in [1..n-1]$. 
For this reason, for a given $P_A$ and $P_N$, it
holds that $\overline{D}(q,P_N) = 0$ and $D(q, P_N) = 0$ for every point $q \in \Gamma(P_A)$.  In the following, we discuss particular properties that arise from this ordering and that will lead to an
incremental evaluation of both lower bounds. 

For a given $P_A$ and $P_N$,
let $p \in P_N$ be the smallest point with respect to the ordering
above. 
We consider the following subsets of $\overline \Gamma(P_A \cup \{p\})$:
$$ \overline G_0(p, P_A) := \overline \Gamma (P_A \cup \{p\}) \setminus \overline \Gamma (P_A) \text{\qquad and \qquad }
\overline G_1(p, P_A) := \left\{q \in \overline \Gamma(P_A) \mid p < q \right\}. $$
Note that due to the ordering of the points in $P$, we have that 
$\overline{G}_1(p,P_A)$ contains only points in $\overline{\Gamma}(P_A)$
that strictly dominate point $p$ and 
have $1$ in the first coordinate.

\paragraph{Update of $LB_1$} 
We state the following propositions for the incremental computation of 
$LB_1(P_A,P_R,P_N)$ in the case of 
inserting point $p$ into $P_A$ and into $P_R$, respectively.
\begin{proposition}
	\label{prop:1}
\begin{eqnarray*}
	LB_1(P_A \cup \{p\},P_R, P_N \setminus \{p\}) & = & \max 
		\begin{cases}
		LB_1(P_A,P_R,P_N) \\
		\displaystyle
		\max_{q \in \overline G_0(p, P_A)}  
		\left\{\eta(q, P_A \cup \{p\}, P_N \setminus \{p\})\right\}
		\end{cases}\\
\end{eqnarray*}
\end{proposition}
\begin{proof}
	We prove that for every point $q$ in $\overline \Gamma(P_A)$,
	it holds that $\eta(q, P_A, P_N) = \eta(q, P_A \cup \{p\}, P_N \setminus \{p\})$
	and, therefore, only the points in $\overline {G}_0(p, P_A)$ need to be
	considered. For this, we partition  $\overline \Gamma(P_A)$ in two 
	disjoint subsets, $\overline \Gamma(P_A) \setminus \overline G_1(p, P_A)$
	and $\overline G_1(p, P_A)$.
\begin{itemize}
	\item [i)] If $q \in \overline \Gamma(P_A) \setminus \overline G_1(p, P_A)$, then  
	$D(q, P_A \cup \{p\}) = D(q, P_A)$ and $D(q, P_N
		\setminus \{p\}) = D(q, P_N)$.  
	\item [ii)] If $q \in \overline
		G_1(p, P_A)$, then $D(q, P_A \cup \{p\}) = D(q, P_A)+1$ and
		$D(q, P_N \setminus \{p\}) = D(q, P_N) - 1$, and thus 
		$\min\{m,  D(q, P_A \cup \{p\}) +  D(q, P_N \setminus \{p\}) \}
		= \min\{m, D(q, P_A)+D(q, P_N)\}
		$
		.
\end{itemize}
\end{proof}

\begin{proposition}
\begin{eqnarray*}
	LB_1(P_A, P_R \cup \{p\}, P_N \setminus \{p\}) & = & \max 
		\begin{cases}
		LB_1(P_A,P_R,P_N) \\
		\displaystyle
		\max_{q \in \overline G_1(p,P_A)}  
		\left\{\eta(q, P_A, P_N \setminus \{p\})\right\}
		\end{cases}\\
\end{eqnarray*}
\end{proposition}

\begin{proof}	
	We prove that for every point $q$ in $\overline \Gamma(P_A) \setminus 
	\overline G_1(q, P_A)$,
	it holds that $\eta(q, P_A, P_N) = \eta(q, P_A, P_N \setminus \{p\})$,
	and therefore, only the points in $\overline G_1(q,P_A)$ need to be 
	considered. The proof is similar to part i) of the proof 
	of the Proposition~\ref{prop:1}, except that only $D(q,P_N)$ and 
	$D(q,P_N \setminus \{p\})$ are taken into account.
If $q \in \overline \Gamma(P_A) \setminus \overline G_1(p, P_A)$, then 
$D(q, P_N \setminus \{p\}) = D(q, P_N)$.  
\end{proof}
\paragraph{Update of $LB_2$} 
For the second lower bound computation, 
we consider the following subset of $\Gamma(P_A \cup \{p\})$:
$$ G_0(p, P_A) := \Gamma (P_A \cup \{p\}) \setminus \Gamma (P_A).$$ 
We state the following equalities.
\begin{proposition}
\begin{eqnarray*}
	LB_2(P_A \cup \{p\},P_R, P_N \setminus \{p\}) & = & \max 
		\begin{cases}
		LB_2(P_A,P_R,P_N) \\
		\displaystyle
		\max_{q \in G_0(p,P_A)} \left\{ \frac1m \overline{D}(q,P_A \cup \{p\}) - \lambda(q)\right\}
		\end{cases}\\
\end{eqnarray*}

\end{proposition}

\begin{proof}
	Similar to the proof of Proposition \ref{prop:1}. If 
	$q \in \Gamma(P_A)$, then we have that $\overline D(q, P_A \cup \{p\}) = \overline D(q, P_A)$ and $\overline D(q, P_N
		\setminus \{p\}) = \overline D(q, P_N)$.   
\end{proof}

The following proposition simply uses the fact that $LB_2$ is only defined via $P_A$. Moving a point from $P_N$ to $P_R$ does not have any effect on the value of this lower bound.  
\begin{proposition}
It holds that $LB_2(P_A, P_R \cup \{p\}, P_N \setminus \{p\}) = LB_2(P_A,P_R,P_N)$.
\end{proposition}

The sets $\overline{G}_0(p,P_A)$ and $G_0(p,P_A)$ are of size $O(dm^{d-1})$ at worse, as we have $d$ possible choices for a coordinate taken from $p$ and, given this fixed coordinate, we have $(|P_A|+1) \leq m$ choices for each of the other $d-1$ coordinates. The set $\overline{G}_1(p,P_A)$ is of size $O(m^{d-1})$ as the first coordinate is fixed and we could have up to $|P_A| \leq m$ choices for each of the other coordinates in the worst case.
The results above indicate that the lower bound can be computed
incrementally in $O(dm^{d-1})$ time at each recursive step, assuming that 
$D(q,P_N)$ and $\overline{D}(q^\prime,P_N)$ can be computed in constant time
after a pre-processing step as discussed in this section. 

\subsection{Greedy Heuristic}
\label{sec:greedy}

An initial upper bound for BB is given by a greedy heuristic that selects $m$ points
iteratively.  The greedy choice consists of selecting the point amongst those
that were not yet chosen that gives the best improvement in terms of star
discrepancy (note here that this improvement can be negative, as discussed in Example~\ref{ex}). 
Therefore, the selection of the next point involves the
evaluation of $O(n)$ star discrepancies, each of which takes 
$O(m^d)$ time with a na\"ive approach. Although better running times can be achieved, we found this
procedure to be reasonably fast for the size of the point sets considered in our 
experimental analysis. 
In the subsequent sections, we will include performance statistics for the greedy heuristic in our reports, to provide an impression for its quality in the various use-cases.

\section{Comparison of the Different Algorithms}
\label{sec:comp-algos}

We have presented above three different strategies to address the discrepancy subset selection problem: an MILP formulation, the branch-and-bound algorithm, and the greedy strategy. In this section we compare the efficiency of these three algorithms. We add to the comparison a na\"ive random sampling approach, which simply selects random subsets of the target size $m$. 

The MILP solver and the branch-and-bound algorithm do not always terminate within the given time limit. In these cases, they can nevertheless report the best solution that they have been able to find. 

\subsection{Experimental setup}

The {\tt Sobol}, {\tt Halton} and {\tt RevHal} point sets were generated by a program written in 
{\tt C} using GNU Scientific Library, namely, library {\tt gsl\_qrng} for the
generation of quasi-random sequences, with the procedures {\tt
gsl\_qrng\_sobol}, {\tt gsl\_qrng\_halton}, and {\tt gsl\_qrng\_reversehalton},
respectively. The sequences {\tt unif} were also generated in a similar way, 
using library {\tt gsl\_rng} for random number generation with the procedure
{\tt gsl\_rng\_uniform}. {\tt Faure} and {\tt iLHS} point sets were generated in
{\tt R} using procedure {\tt runif.faure} available in the {\tt DiceDesign} package
and procedure {\tt improvedLHS} available in the {\tt lhs} package, respectively.
{\tt Fibon} sets were generated by a code in {\tt Python} (version 2.7.16)
written by the authors. 
 
For the two-dimensional case, we considered 
$m \in \{20, 40, 60, 80, 100, 120\}$ and for each value of $m$, we considered $n \in \{m+20, m+40, \ldots, 140\}$. 
For the three dimensional case, we considered 
$m \in \{20, 40, 60, 80\}$ and for  each value of $m$, $n\in \{m+20,m+40, \ldots, 100\}$.
Preliminary experiments indicated that larger values of $n$ would increase the computational cost significantly, requiring several hours of computation time before the algorithms converge. For the two randomized constructions {\tt iLHS} and {\tt unif}, 
we have generated 10 instances for each combination 
of values of $m$ and $n$.

To compare the discrepancy values of the subsets with the original size-$m$ point sets, we also 
computed the discrepancy values of the latter, using the (exact) algorithm described in~\cite{discrepancy_algorithm} and provided to us by Magnus Wahlstr\"om. For consistency, we denote theses cases as ``$n=m$''. 

We used SCIP solver version 7.0.1 to solve the MILP formulation described in  Section~\ref{sec:milp}. The MILP formulation was written in an LP format, which is read and solved by SCIP solver with the default parameters.
The BB algorithm for two and three dimensions and 
with the incremental computation of lower bounds as described in Section \ref{sec:BBalgo} 
was written in {\tt C}. In a preliminary step, the points were 
sorted in increasing order with respect to the first dimension to prepare them for the application of our solvers. 

To run the experiments, we used a computer cluster 
Dell PowerEdge R740 Server with two Intel Xeon Silver 4210R 2.4G, 
10 Cores / 20 Threads, 9.6GT/s, 13.75M cache, with two 32GB RDIMM, 
two 480GB SSD SATA hard-drives, and Debian GNU/Linux 10 (buster) operating 
system. The running times in seconds of the SCIP solver and of the BB program were measured with command {\tt time} under 
linux, with a cut-off time limit of 30 minutes.
The time to generate the files with the MILP formulation were not 
taken into account. For the BB program, we used 
{\tt gcc} compiler version 8.3.0 with {\tt -O3} compilation flag.
We have only used arrays with static memory allocation. 

\subsection{Quality of Random Subset Sampling and the Greedy Heuristic}
\label{sec:random-vs-greedy}

To gain a feeling for the complexity of the subset selection problem, we first study the solution quality of randomly selected subsets of target size $m$ as well as that of the greedy heuristic described in Section~\ref{sec:greedy} (i.e., the strategy used to initialize the upper bound for the BB method). 


\begin{table}[t]
\centering
\begin{tabular}{llllll}
\toprule
quantile      & \tt Faure   & \tt Sobol'  & \tt Halton  & \tt RevHal & \tt Fibon \\
\midrule
best possible subset & 0.0357  & 0.0356 & 0.0359  & 0.0363      & 0.0351    \\
best found subset   & 0.0547  & 0.0540 & 0.0531  & 0.0542      & 0.0518    \\
1\%           & 0.0718  & 0.0715 & 0.0714  & 0.0713      & 0.0688    \\
10\%          & 0.0844  & 0.0836 & 0.0843  & 0.0838      & 0.0808   \\
25\%          & 0.0937  & 0.0928 & 0.0942  & 0.0935      & 0.0899   \\
50\%          & 0.1065  & 0.1055 & 0.1078  & 0.1063      & 0.1025   \\
75\%          & 0.1219  & 0.1212 & 0.1245  & 0.1222      & 0.1177   \\
90\%          & 0.1382  & 0.1370 & 0.1418  & 0.1384      & 0.1334   \\
100\%  & 0.24481 & 0.2459 & 0.2673 & 0.2531    & 0.2478  \\
\bottomrule
\end{tabular}
\caption{Percentiles of the star discrepancy values found by random subset sampling with 1\,000\,000 trials, for the instances with $n=100$ points and subset size $m=60$ in dimension $d=2$.}
\label{tab:quantiles-random}
\end{table}

Table~\ref{tab:quantiles-random} shows selected percentiles of star discrepancy values for 1\,000\,000 i.i.d. uniformly selected subsets of size $m=60$ for the five considered low-discrepancy sequences with $n=100$ points in dimension $d=2$. The distributions are quite similar for the five point sets. The main probability mass is around about twice the solution quality of the best found subset. The latter, in turn, have discrepancy values that are still between 47.5\% and 53.3\% worse than the best possible subset. Even if the evaluation of 1\,000\,000 subsets could be executed in less than two minutes, the results already suggest that random sampling is quite inefficient for the discrepancy subset selection problem. 

That the inefficiency of the random subset sampling is not an artifact of the setting described in Table~\ref{tab:quantiles-random} is indeed confirmed by the values in Tables~\ref{tab:2dsummary} and~\ref{tab:3dsummary} (available in the appendix), where we report, for all tested combinations of $m$ and $n$ in $2d$ and $3d$, respectively, the discrepancy values of the best random subset that could be found within a cut-off time of 30 minutes. For fixed $m$, the values do not significantly improve with increasing~$n$, in contrast to the value of the best possible (or best found) subset of the same size, which are reported in column $subset$. As a result, the relative disadvantage of the random subset selection procedure increases from around 
10\% for $m=20$ and $n=40$ to around 40\% for $n=140$ in the $2d$ case. For $m=40$, the disadvantage is already around 16\% on average for $n=60$ and 58\% for $n=140$. For $m=120$ and $n=140$, the relative disadvantage of random subset sampling is between 22\% for {\tt Sobol} and 32\% for {\tt Halton}. For the $3d$ case, the best subsets found by random sampling are around 19\% worse on average than the optimal ones for $m=20$ and $n=40$, and this value increases to around 30\% for $m=20$ and $n=60$ and to around 35\% for $m=20$ and $n=80$ and $n=100$.  

Comparing random subset sampling to the greedy strategy (column $greedy$ in Tables~\ref{tab:2dsummary} and~\ref{tab:3dsummary} in the appendix), we see that random subset sampling provides much better upper bound; however, we should keep in mind that several thousands of millions of subsets are evaluated within the 30 minutes time limit of the random subset sampling strategy, whereas the greedy strategy is deterministic and therefore evaluates only a single subset. As discussed above, the figures in Table~\ref{tab:quantiles-random} showed that already after two minutes the median performance of random subset sampling was around twice as large as the value of the best found subset, so that the comparison between random subset selection and the greedy strategy should indeed be done with care. 
For fixed $m \in \{20,40,60\}$, the greedy strategy tends to give worse solutions when the number of available points, $n$, increases. Across all evaluated settings, its discrepancy values are between 33\% and 172\% worse than the best (or best found) subset, with an average overhead of 92\% and a median of 88\%. The average and the median disadvantage of the greedy strategy compared to the result of the random subset sampling are both around 40\%.   

Most observations made for the low-discrepancy sequences carry over to the performance on the subset selection problem on {\tt iLHS} and {\tt unif}, as can be seen in Tables~\ref{tab:2dsummary_lhs} and~\ref{tab:2dsummary_unif} for the $2d$ case and in Tables~\ref{tab:3dsummary_lhs} and~\ref{tab:3dsummary_unif} for the $3d$ case, respectively, in the appendix. 
In particular, the performance of $random$ subset sampling decreases with increasing $n$ and fixed $m$, whereas the values of the best possible subsets improve. In fact, not only the relative but also the absolute value of the best subset found by random subset sampling increases with increasing $n$, and this consistently for all $m$ in $2d$ in the case of {\tt iLHS} point sets and for most values of $m$ in the {\tt unif} case (no correlation between different values of $n$ can be identified for the case $m=20$ nor in the two cases for the $3d$ setting). No clear correlation between the quality of the greedy strategy and the value of $m$ and $n$ can be identified, except that for the $2d$ {\tt iLHS} samples the absolute values of the subset computed for $n=140$ tend to be worse than those for smaller $n$. This effect, however, cannot be observed in the $2d$ {\tt unif} nor in the $3d$ cases. 

The bounds provided by the $greedy$ strategy are quite stable for varying $n$ and fixed $m$, but are significantly worse than bounds provided by the $random$ strategy in $2d$. In $3d$, however, this is not the case. Here, the results of the $greedy$ strategy are better than those of the random subset sampling; the average (median, max) advantage of the greedy strategy over the random one is 20\% (30\%, 38\%) for {\tt iLHS} and 16\% (18\%, 36\%) for {\tt unif}. In some of the $3d$ cases, the best value returned by the greedy strategy is either optimal (this is the  case for the $m=60$, $n=80$, {\tt iLHS} setting) or could not be improved by the exact solvers ($m=80$, $n=100$, {\tt iLHS}; $m=60$, $n=80$, {\tt unif}; and $m=80$, $n=100$, {\tt unif})  or it is quite close to optimal (e.g., $m=80$, $n=100$, {\tt iLHS} with a 0.5\% overhead compared to the best value returned by the exact solvers). Note that the $random$ strategy evaluates much fewer samples in $3d$ than in $2d$, since the star discrepancy computation is substantially more time-consuming in $3d$.  

Thus, summarizing this section, we find that (with few exceptions), both the $random$ subset sampling and the $greedy$ heuristic perform rather poorly on the discrepancy subset selection problem, clearly motivating the need for more sophisticated approaches, either in terms of exact solvers such as the MILP and BB approaches presented in Section~\ref{sec:algos} or in terms of better heuristics.      

\subsection{Comparison between MILPs and Branch-and-Bound}
\label{sec:milp-vs-bb}

\begin{table}[t]
\centering
        {\footnotesize
		\begin{tabular}{l@{  }lr@{ }rr@{ }rr@{ }r}\hline
		\multirow{1}{*}{$m$} & \multirow{1}{*}{sequence}
		            & \multicolumn{2}{c}{$n=40$}  & \multicolumn{2}{c}{$n=60$} 
			    & \multicolumn{2}{c}{$n=80$}\\\toprule
$20$ & \tt Faure 	   &34 &&859   &&  - &\\
     & \tt Sobol' 	   &4  &&973   &&  - &   \\  
     & \tt Halton 	   &30 && -    &&  - &  \\  
     & \tt RevHal    &10  &&1278  &&1378&  \\
     & \tt iLHS        &161 &(9)& 620 & (8)& 567 & (1) \\
     & \tt unif       &31 &(9) & 305 & (8)& 966 & (1)\\\midrule
$40$ & \tt Faure 	   &&  &-&   &-&   \\ 
     & \tt Sobol' 	   &&  &-&   &-&   \\
     & \tt Halton 	   &&  &-&  &-&   \\
     & \tt RevHal   &&  &-&  &-&   \\
     & \tt iLHS        &  & & 253 & (2)& - &  \\
     & \tt unif       &  & &  - & & - &  \\\midrule
$60$ & \tt Faure 	   &&  &&  &-&  \\
     & \tt Sobol' 	   &&  &&  &-&  \\
     & \tt Halton 	   &&  &&  &-&  \\ 
     & \tt RevHal     &&  &&  &-& \\
     & \tt iLHS        &  & &   & & 806 &(2) \\
     & \tt unif       &  & &   & & 86  &(2) \\\bottomrule

	\end{tabular}}
	\caption{CPU-time (in seconds) of 
	BB for low-discrepancy sequences 
	and median CPU-time and number of instances solved
	out of ten (in parenthesis) 
	for randomized constructions
	for several values of $n$ and $m$ 
	in the three-dimensional case, where ``-'' indicates
that the approach did not terminate before the time limit of 1800 seconds.}
\label{tab:CPU-3d}
\end{table}

Tables~\ref{tab:CPU-2d} and~\ref{CPU-2d-lhs} in the appendix present the running times (measured in seconds) of the MILP solver and BB, for different values of $n$ and $m$ in dimension $d=2$ on deterministic sequences and on randomized constructions, respectively. The same information for branch and bound in $d=3$ is shown in Table~\ref{tab:CPU-3d} (for the randomized constructions more details can be found in Table~\ref{CPU-3d-lhs} in the appendix). We do not show the results for the MILP solver, since they have shown poor performance in the $3d$ case. For the randomized sequences, the values reported in Table~\ref{tab:CPU-3d} are the median running times on the instances that were solved within the cut-off time. 
In Tables~\ref{CPU-3d-lhs} and~\ref{CPU-2d-lhs}, we report
the minimum, the median and the maximum value. The number of instances that were solved 
within the time limit is reported in parenthesis. The entry ``-'' indicates that the implementation was not able to terminate within this cut-off time.

For a better comprehension of the information shown in Tables \ref{tab:CPU-2d} and~\ref{CPU-2d-lhs}, Figure~\ref{fig:2Druntimes} plots a summary of these tables.
The left and the right column correspond to the performance of the ILP solver 
on the MILP formulation and of BB, respectively. 
Each row corresponds to the performance obtained for a given value of~$m$.
The points correspond to the running times in seconds obtained on 
deterministic sequences 
(fau is {\tt Faure}, sob is {\tt Sobol} hal is {\tt Halton}, rev is {\tt RevHal}, and fib is {\tt Fibon}) and to the median running times in seconds on 
randomized constructions (lhs is {\tt iLHS} and uni is {\tt unif}). The label for each deterministic sequence is placed close to the point with the largest value of
$n$ for which the approach was able to solve before the time limit of 1800 seconds was achieved. In the case of randomized constructions, 
the label is placed close to the point with the largest value of $n$ for which the approach was 
able to solve at least one instance before the time limit. 
The barplots present the percentage of solved instances of randomized 
constructions, where violet and blue correspond to {\tt iLHS} and {\tt unif}, respectively.

The results for $d=2$ in Tables \ref{tab:CPU-2d} and~\ref{CPU-2d-lhs}
suggest two different patterns
for MILP and for BB: while the latter is faster to find the optimal subset for small
$m/n$ ratios (see row $m=20$ in Table~\ref{tab:CPU-2d} and second column and
first row in Figure~\ref{fig:2Druntimes}), MILP is faster for
larger $m/n$ ratios (see main diagonal in Table~\ref{tab:CPU-2d} and the
leftmost running-times in the left column of Figure~\ref{fig:2Druntimes}). The
difference between the two methods is striking at both ends. BB solves 
all instances with $n=140$ and $m=20$ in almost less than 100 seconds whereas MILP
cannot even solve a single one. MILP can solve all instances for $n=140$ and
$m=120$, except for {\tt Fibon} sequences, whereas BB can only solve {\tt
RevHal}, {\tt Fibon}, and almost half of the {\tt iLHS} instances. 

The strong performance of MILP as compared with BB when the $m/n$ ratio is close to 1 is related to the quality of the lower bounds on those cases. Table
\ref{tab:2dLP} in the appendix shows the integrality gap of the LP relaxation of MILP for
$2d$ deterministic sequences with respect to the optimal found and,
when not available, with respect to the best solution found. We can observe
that the smallest gaps arise for larger $m/n$ ratios. We recall that the solution
of the LP relaxation of MILP is integral when $m=n$. 

Unfortunately, MILP is not feasible for the $3d$ case
as the memory requirement grows very fast, reaching the limit
available in the cluster for the smallest instances. This is due to the large number of constraints in 
the $3d$ case. In fact, a 
file with the MILP formulation in LP format occupies several 
gigabytes. For this reason, Tables~\ref{tab:CPU-3d} and~\ref{CPU-3d-lhs} report only the CPU-time taken by BB on deterministic and randomized sequences, 
respectively. The performance decay of BB 
reported in these tables is noticeable in comparison with the $2d$ case. For instance, this approach cannot find a 
single solution for $3d$ instances with $n=100$ and $m=20$, whereas it can solve all $2d$ instances 
for the same values of $m$ and $n$ in at most 250 seconds. This is mainly due to the time taken 
with the update of data structures and the evaluation of the 
lower bounds, which grows considerably with an increasing number of dimensions. 

We also observe from Table~\ref{tab:CPU-2d}  that there seems to be little
difference of performance among the different deterministic
sequences. Still, some outliers are quite noticeable with BB,
such as with {\tt Faure} sequence for $n=100$ and $m=80$ which took 1\,194
seconds, while with {\tt Fibon}
sequence for the same parameters took only 45 seconds. Similarly, 
MILP was not able to solve a {\tt Fibon} sequence for $n=120$ and $m=100$,
took 1\,538 seconds to solve a {\tt Sobol} sequence with the same sizes, 
while it took less than 20 seconds to solve the remaining sequences. 
There is also no large difference between the running times obtained
on deterministic sequences and on {\tt iLHS} sequences. Differently, we
observe that {\tt unif} sequences take less time to be solved with MILP 
while they take more time to be solved with BB.

\begin{table}[t]
\centering
        {\footnotesize
                \begin{tabular}{llrr}\hline
        \multirow{2}{*}{$m$} & \multirow{2}{*}{sequence} &  \multicolumn{2}{c}{$n=140$}\\
             &                  & MILP & BB  \\\hline
        $40$ & \tt Faure            &*0.0449& 0.0449 \\      
             & \tt Sobol'            &*\textcolor{gray}{0.0449}& 0.0447 \\
             & \tt Halton           &*0.0452& 0.0452 \\
             & \tt RevHal        &0.0444 & 0.0444 \\
             & \tt Fibon        & 0.0448  & 0.0448 \\\hline

        $60$ & \tt Faure            & *0.0334 &  0.0334 \\          
             & \tt Sobol'            & *\textcolor{gray}{0.0334} &  0.0328 \\
             & \tt Halton           & *\textcolor{gray}{0.0345} &  0.0338 \\
             & \tt RevHal        & *0.0336 & 0.0336 \\
             & \tt Fibon        &*\textcolor{gray}{0.0338}   & 0.0338 \\\hline

        $80$ & \tt Faure            & *\textcolor{gray}{0.0273} &  0.0271 \\
             & \tt Sobol'            & *0.0273 &  0.0273 \\
             & \tt Halton           &  0.0277 &  0.0277 \\
             & \tt RevHal        & *\textcolor{gray}{0.0279} & *0.0278 \\
             & \tt Fibon        & *0.0296  &*0.0276  \\\hline

        $100$ & \tt Faure           & 0.0241 & 0.0241 \\
              & \tt Sobol'           & 0.0241 &*\textcolor{gray}{0.0246} \\
              & \tt Halton          &*0.0242 &*\textcolor{gray}{0.0297} \\
              & \tt RevHal       & 0.0238 &*0.0238 \\
              & \tt Fibon        & *0.0296  & 0.0230\\\hline

        \end{tabular}}
        \caption{Optimal and best found (*) star discrepancy values for MILP and BB  with a time limit of 23 hours. All data is for the two-dimensional case with $n=140$ and different values of $m$. Provably non-optimal values are printed in \textcolor{gray}{grey color.}}
\label{tab:mm2-vs-bb-d2-n140}
\end{table}

Table~\ref{tab:mm2-vs-bb-d2-n140} compares the final solution quality of the
MILP and BB for $n=140$
and different values of $m$, after a cut-off time of 23 hours. Values marked by
an asterix * could not be proven to be optimal by the solver, and values
printed in grey color are known to be non-optimal, by the result of the other
solver. BB could solve all but four instances, whereas the solver for
MILP did not finish on eleven instances. The solution quality, however, is
nevertheless decent: only two values deviate from the best solution found by
BB by more than $1\%$; these are for $m=60$ and {\tt Sobol} (1.8\%
worse than the optimal solution) and for {\tt Halton} (2.1\% worse). 
Only for the case of $m=100$ for {\tt Sobol} sequence 
(+2.1\% compared to the optimal solution) and
for {\tt Halton} (+22.7\% compared to the best solution found by MILP) are the values obtained by BB worse than those obtained by MILP.

\begin{figure}
    \centering
    \includegraphics[width=0.97\textwidth]{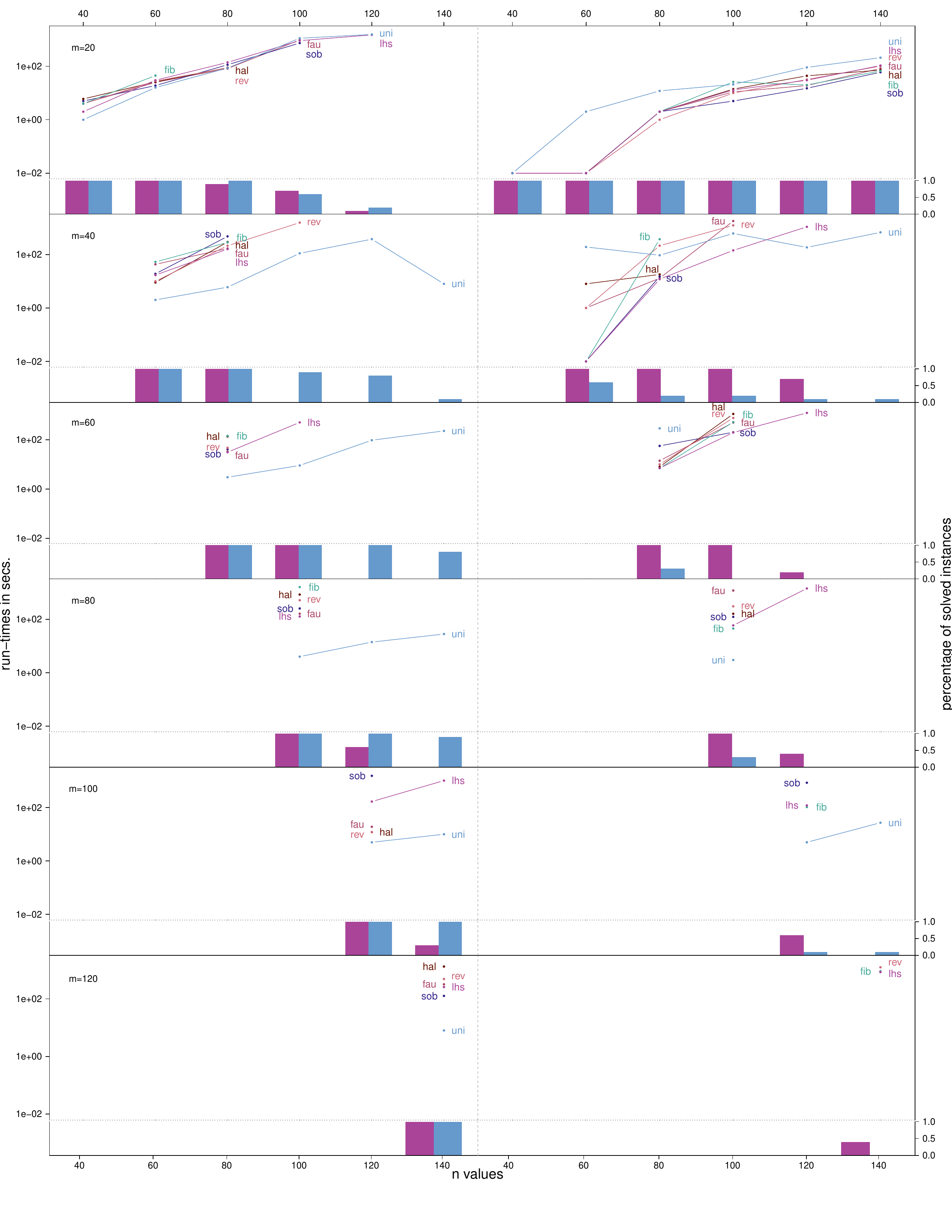}
    \caption{Run-times for deterministic sequences, median run-times for randomized sequences (points and lines), and percentage of solved instances for randomized sequences (barplots) 
    for MILP (left column) and BB (right column) and for each combination of $m$ (rows) and $n$.}
    \label{fig:2Druntimes}
\end{figure}

\section{Comparison of Star Discrepancy Values}
\label{sec:comparison-values}\label{sec:comp-discrepancy}

While we have focused in Section~\ref{sec:comp-algos} on the comparison between the different solvers, we now discuss the quality of the subsets for the different point constructions. Detailed values and information about the convergence of the exact solvers can be found in Tables~\ref{tab:2dsummary}, \ref{tab:2dsummary_lhs}, and~\ref{tab:2dsummary_unif} for low-discrepancy, {\tt iLHS}, and {\tt unif} samples in $2d$ and in Tables~\ref{tab:3dsummary}, \ref{tab:3dsummary_lhs}, and~\ref{tab:3dsummary_unif} for low-discrepancy, {\tt iLHS}, and {\tt unif} samples in $3d$, respectively. 

\begin{figure}[t]
    \centering
    \includegraphics[width=\textwidth]{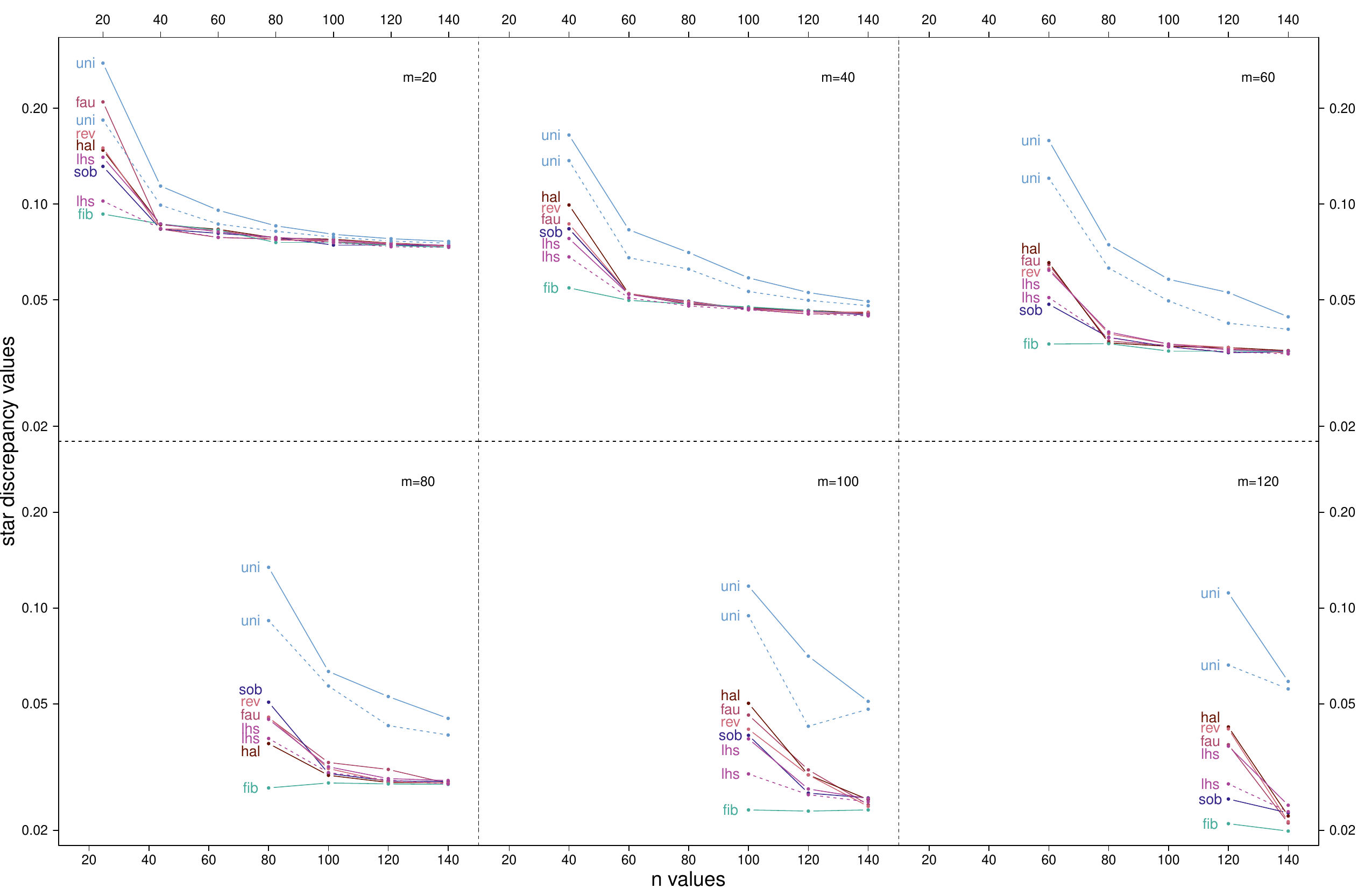}
    \caption{Star discrepancy values for each tested combination of
    $m$ and $n$ in $2d$. For the two randomized constructions {\tt iLHS} and {\tt unif}, minimum (dashed lines) and median (solid lines) values across the ten independent runs are shown.}
\label{fig:comp-disc}
\end{figure}

\subsection{The Two-Dimensional Case} 
Figure~\ref{fig:comp-disc} visualizes the star discrepancy values of the optimal (or best found, see Tables~\ref{tab:2dsummary}, \ref{tab:2dsummary_lhs}, and~\ref{tab:2dsummary_unif} for details) subsets for all tested combinations of $n$ and $m$ in $2d$. 

\begin{figure}
    \centering  
    \includegraphics[width=\linewidth, trim=1.7cm 12.5cm 2cm 2cm, clip]{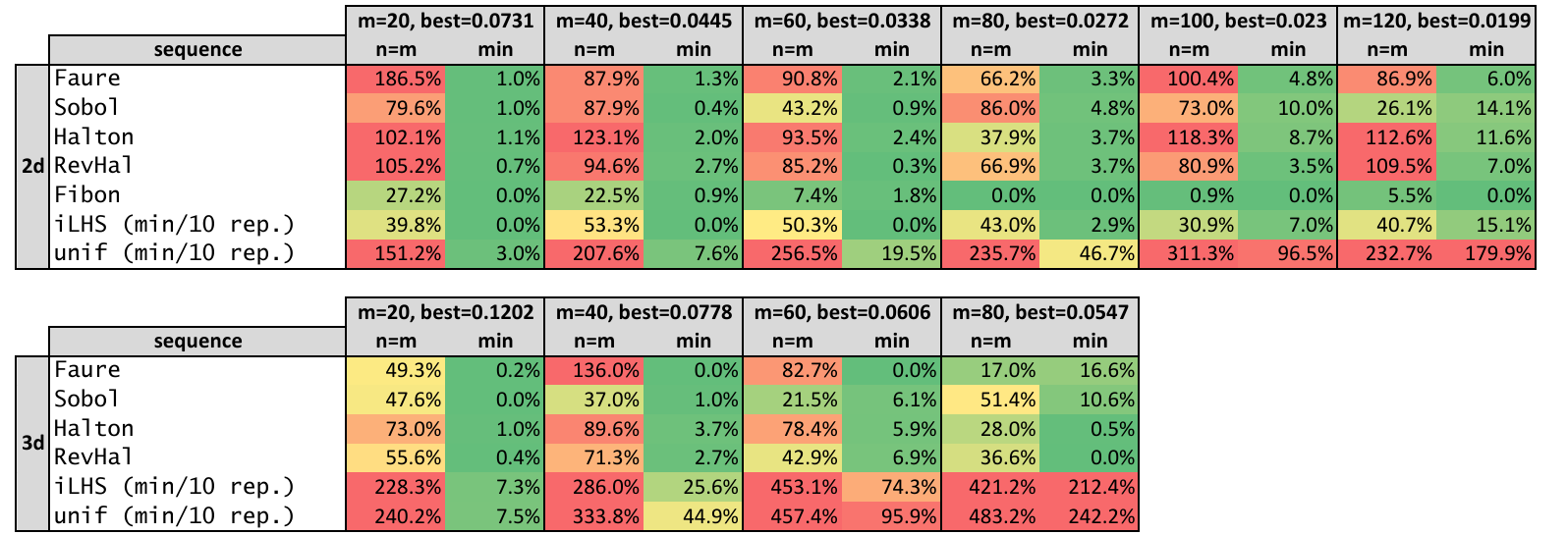}
    \caption{
    Relative disadvantage of the discrepancy values of the original point sets (column ``$m=n$'') and of the best size-$m$ subset (across all tested sets with $n=m+20i$ points, column ``min''), compared against the best overall set with $m$ points. For the two random constructions,{\tt iLHS} and {\tt unif}, we report the best out of the ten independent experiments.}
    \label{fig:overhead}
\end{figure}

\emph{Dependency on $m$.} 
As expected, the discrepancy values decrease with increasing $m$. While the discrepancy of the best original construction with $m$ points decreases from 
0.093 to 0.0545, 0.0363, 0.0272, 0.0232, and 0.021 for $m=20$, $40$, \ldots, $120$ points, the discrepancy value of the best found size-$m$ subset (over all $n>m$ studied) decreases from 0.0731 for $m=20$ to 0.0445 for $m=40$, 0.0338 for $m=60$, 0.0272 for $m=80$, 0.023 for $m=100$, and 0.0199 for $m=120$. The advantage of the subset selection is therefore around $21\%$ for $m=20$, 18\% for $m=40$, 7\% for $m=60$,	0\% for $m=80$, 1\% for $m=100$, and 5\% for $m=120$. 

\emph{Dependency on $n$.} For fixed $m$, the values tend to decrease with increasing $n$, but there are a few cases that do not follow this rule, i.e., in which the optimal $m$ point subset of a $n=m+20i$ set has greater discrepancy value than the set with $n=m+20(i-1)$ points. Cases with $n=140$ may be caused by non-convergence of the exact solvers, i.e., the reported bounds may simply not reflect the value of an optimal subset. Examples for this setting are {\tt Faure} with $m=40$, {\tt Sobol} with $m=60$, {\tt Fibon} with $m=100$. However, there are also cases in which the increase in discrepancy value is not caused by this artifact, but by a real disadvantage of the larger $n$-point set. This is the case for the {\tt Fibon} sequence with $m=60$, where the discrepancy of the optimal subset of the $n=80$ construction is $0.0364$, slightly larger than the $0.0363$ discrepancy of the original $m=60$ construction. It is also the case for the {\tt Fibon} sequence with $m=80$, which has a discrepancy value of $0.0272$ for the original ($n=m$) construction, whereas the optimal subset of the $n=100$ point set has discrepancy $0.0282$ (and also the best found subset for the $n=120$ construction is worse than the original 80-point one, but the solvers did not converge, so that we do not know whether the disadvantage is real). Another example of a non-monotonic behavior is the {\tt unif} construction with $m=100$, but here the decrease in the discrepancy value of the best subset is simply caused by the random nature of the construction, and the comparatively large variance between the independently sampled $n$-point sets, see Figure~\ref{fig:cbox} for an illustration. 

We also observe a general trend for diminishing returns for increasing $n$, i.e., the relative gain when increasing $n$ from $m$ to $m+20$ is larger than the gain when increasing $n$ from $m+20i$ to $m+20(i+1)$ for~$i>0$. 

\begin{figure}[t]
    \centering
    \includegraphics[width=\textwidth]{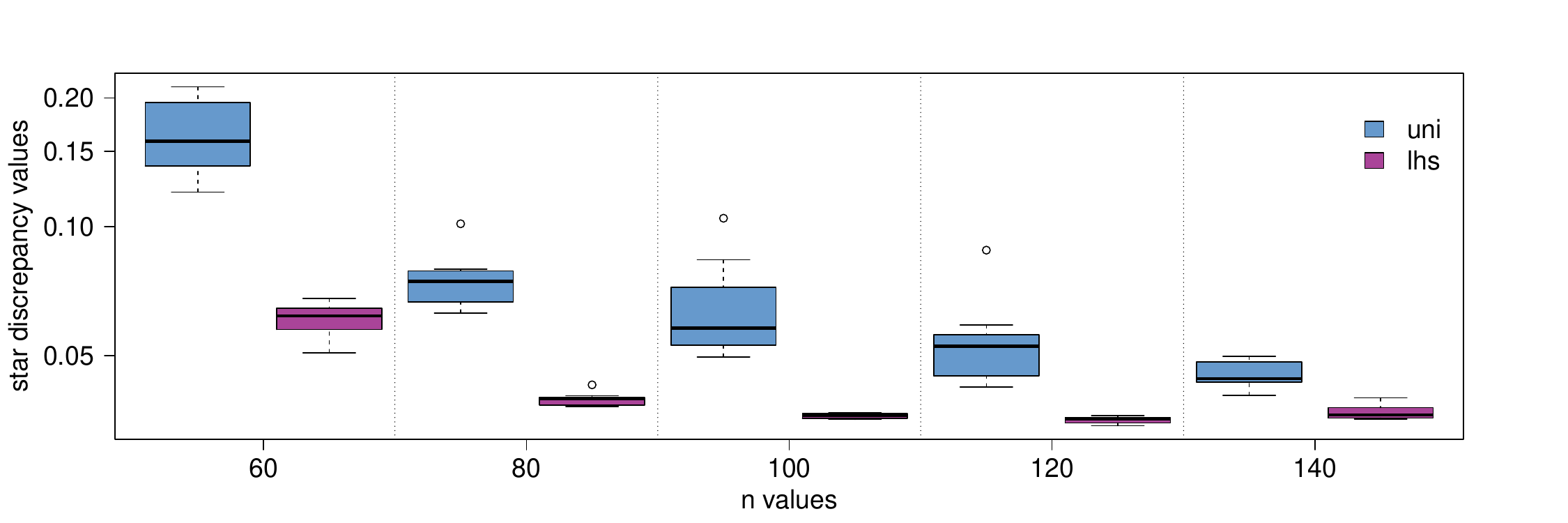}
    \caption{Boxplot of the star discrepancy values of 
    the best sets found for {\tt iLHS} and {\tt unif} in
    $2d$ for $m=60$ and for several
    values of $n$.}
\label{fig:cbox}
\end{figure}

\emph{Comparison of the different constructions.} 
The by far worst discrepancy values are obtained by the uniformly sampled point sets {\tt unif}, and this even when considering the best of all ten independent runs (dashed line in Figure~\ref{fig:comp-disc}). For the original $m$-point constructions, i.e., the case $n=m$, the {\tt Fibon} sets are clearly the best, with discrepancy values that are significantly smaller than that of all other constructions. However, we also see that the advantage of this sequence diminishes or even vanishes when considering the best size-$m$ subsets that could be identified for $n>m$. Indeed, we observe that the discrepancy values of the $n=m$ point sets can differ quite substantially between the different constructions, whereas their values are quite similar for the best (found) size-$m$ subsets out of the $n=140$ constructions.\footnote{We note that in Table~\ref{tab:2dsummary} the discrepancy of the best found size $m=100$ subset of the first $n=120$ elements of the {\tt Fibon} sequence is smaller than that of the best size $m=100$ subset found for the first $n=140$ elements. Since the {\tt Fibon} sequence is a proper \emph{sequence,} the first $n=120$ points are contained in the $n=140$ points. This result is hence only possible because our algorithm did not compute the optimal subset for the case $n=140$, as indicated in the table by the *.}   
To analyze these values in more detail, we report in Figure~\ref{fig:overhead} the smallest discrepancy value $d^*_{\infty}(P^*_m)$ found for any of the size-$m$ point sets (top row, value reported as ``best=''). We then report in Figure~\ref{fig:overhead} the relative disadvantage $(d^*_{\infty}(P)-d^*_{\infty}(P^*_m))/d^*_{\infty}(P^*_m)$ of the discrepancy value of the original $m$-point constructions (columns ``$n=m$'') and of the best found size-$m$ subsets (column ``min'') against these best discrepancy values. The disadvantage of the original size-$m$ constructions against the best found point set are quite significant for almost all constructions, with the exception of the {\tt Fibon} sequence, for which the advantage of the subset selection approach varies only between 0\% (for $m=80$) and 27.2\% (for $m=20$). For all other constructions, we see a substantial advantage of the subset selection approach. 

Taking the uniformly sampled point sets aside, the differences between the best size-$m$ subsets are at most 1.1\% for the case $m=20$, at most $2.7\%$ for the case $m=40$, etc. These values increase with increasing $m$, but the plots in Figure~\ref{fig:comp-disc} suggest that a further increase in $n$ could reduce these differences. While the convergence itself may not be very surprising, it is interesting to see that a relatively small increase in $n$ can suffice to find small discrepancy subsets in any of the low-discrepancy construction and in the {\tt iLHS} sets. For uniformly sampled points, larger sample size $n$ seems to be needed to achieve similarly small discrepancy values.
 

For the random point sets, the differences between the discrepancy values of the ten independent {\tt unif} constructions are larger than those of the {\tt iLHS} (sub-)sets, as can be easily seen from the examples plotted in Figure~\ref{fig:cbox} and from the detailed values in Tables~\ref{tab:2dsummary_unif} and~\ref{tab:2dsummary_lhs}.

\subsection{The Three-Dimensional Case} 
Figure~\ref{fig:comp-disc-3d} compares the discrepancy values of the best (found) subsets of size $m$, for all tested super-sets of size $n$. Exact values of the best size-$m$ point set and the relative disadvantages of the six considered constructions are provided on the bottom part of Figure~\ref{fig:overhead}, whereas detailed results are available in Tables~\ref{tab:3dsummary} for the low-discrepancy sequences, \ref{tab:3dsummary_lhs} for {\tt iLHS}, and~\ref{tab:3dsummary_unif} for {\tt unif}, respectively.   

\begin{figure}[t]
    \centering  
    \includegraphics[width=\linewidth] {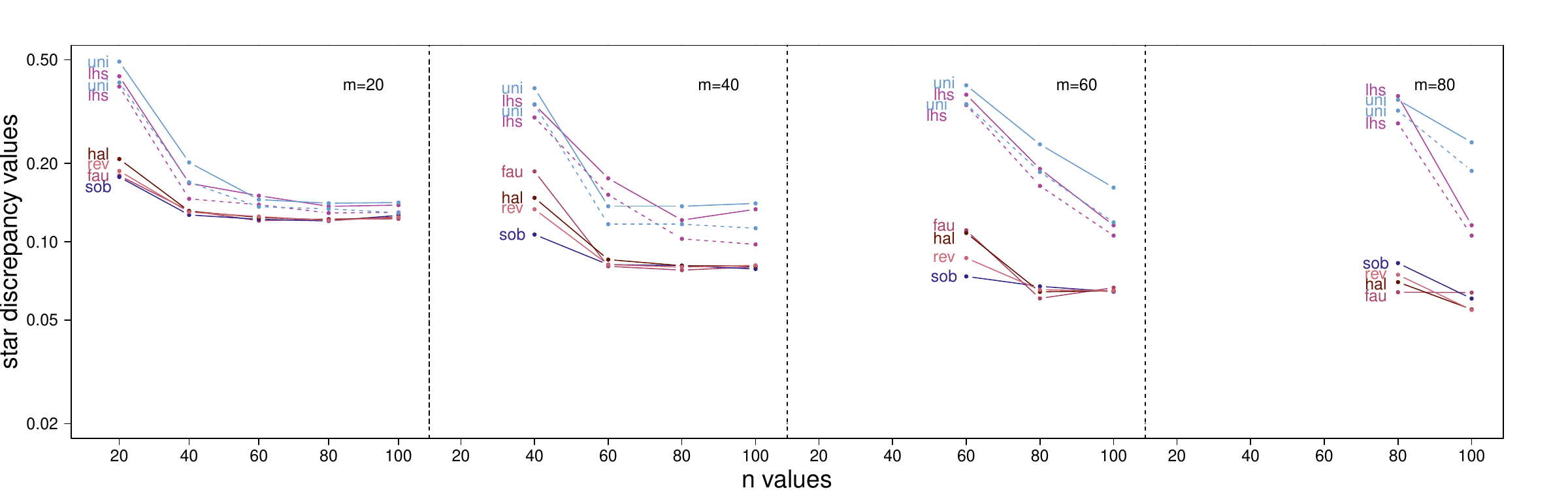}
    \caption{Star discrepancy values for each tested combination of
    $m$ and $n$ in $3d$. For the two randomized constructions {\tt iLHS} and {\tt unif}, minimum (dashed lines) and median (solid lines) values across the ten independent runs are shown.}
\label{fig:comp-disc-3d}
\end{figure}

\emph{Comparison with the $2d$ values.} 
A sequence that clearly stands out in the $2d$ case is the {\tt fibon} sequence. This sequence, however, does not have a straightforward generalization to dimensions $d>2$. It therefore doesn't appear in our $3d$ evaluations. Not surprisingly, the discrepancy values of the $3d$ constructions are much worse than that of the $2d$ constructions for any given $m$.    
For $m=20$, the discrepancy of the best $3d$ set is 64\% larger than that of the best $2d$ set of the same size. This disadvantage monotonically increases with $m$. It is 75\% for $m=40$, 79\% for $m=60$, and 101\% for $m=80$. 

\emph{Dependency on $m$.} As in the $2d$ case, the discrepancy values of the best found size-$m$ point sets decrease with increasing $m$; they are 
0.1202, 0.0778, 0.0606, and 0.0547 for $m=20$, $40$, $60$, and $80$, respectively. That is, the advantage of adding another 20 points decreases with increasing $m$. The discrepancy value of the best original ($n=m$) constructions are 0.1774, 0.1066, 0.0736, 0.064 for $m=20$, $40$, $60$, and $80$, respectively, resulting in a relative advantage of the best size-$m$ subsets over the original ($n=m$) constructions decreasing from 32\% to 27\%, 18\%, and	15\%, respectively.

\emph{Dependency on $n$.} For fixed $m$, the discrepancy decreases with increasing $n$, and this quite significantly already for $n=m+20$, with an average gain of 44\% in discrepancy value for $m=20$ and $m=40$, 36\% for $m=60$, and 26\% for $m=80$. The latter value are based on incomplete data, however, since the algorithms did not converge in the 30 minutes time-out and therefore provided only upper bounds for the discrepancy values of the optimal subset. Based on the same data, the median gain in discrepancy values for $m=20$, $40$, $60$, and $80$ is 34\%, 46\%, 43\%, and 27\%, respectively. 
The values in Figure~\ref{fig:overhead} and the curves in Figure~\ref{fig:comp-disc-3d} show that the advantage is slightly larger when considering the best subset across all tested $n$ values, but -- as in the $2d$ case -- the advantage of increasing $n$ from $m+20i$ to $m+20(i+1)$ decreases rapidly for $i>0$.   

\emph{Comparison of the different constructions.} Comparing the different sequences, we observe a clear disadvantage of the {\tt unif} and the {\tt iLHS} constructions. Even the best found size-$m$ (sub-)sets have a relative overhead of more than 7\% for $m=20$, more than $25\%$ for $m=40$, more than 74\% for $m=60$ and more than $200\%$ for $m=80$. The discrepancy values of the original $m$-point constructions (column $n=m$ in Figure~\ref{fig:overhead}) is above 200\% for all settings, and it is even larger than 400\% for $m=60$ and $m=80$. We recall that ten independently sampled constructions were evaluated, and the values reported in Figure~\ref{fig:overhead} are for the best among these ten trials; the median and average would hence compare even more unfavorably (see Tables~\ref{tab:3dsummary_lhs} and~\ref{tab:3dsummary_unif} for details).    

The differences between the four low-discrepancy sequences are quite small for the best size-20 point set, with less than 1\% difference. For $m=40$, the largest difference between these sequences is 3.7\%. For $m=60$, the {\tt Faure} sequence yields the best subset, and the best subsets of the other sequences are between 5.9\% and 6.9\% worse. For $m=80$, the {\tt RevHal} sequence has the best subset, closely followed by {\tt Halton}. The best {\tt Sobol} and {\tt Faure} subsets are 10.6\% and 16.6\% worse. We suspect that the differences would decrease with increasing $n$, but we could not verify this assumption, since our algorithms did not converge for larger values of $n$. 

Thus, overall, the low-discrepancy sequences show significant advantages over the two randomized constructions, but we do not see any clear ranking of these four tested sequences. For all point sets, the best size-$m$ subsets have significantly smaller discrepancy than the original constructions with $m$ points.  

\section{Conclusions and Future Work}
\label{sec:conclusions}

We have introduced the star discrepancy subset selection problem and we have presented two different exact solvers, one based on mixed-integer linear programming (MILP) and one based on branch and bound (BB). We have compared the performances of these solvers, and contrasted them with that of random subset sampling and a greedy construction. For 
the two-dimensional case, while the MILP solver is efficient for large $m/n$ ratios, BB seems more suitable for small $m/n$ ratios. We relate these findings with the quality of the lower bounds. However, for the three-dimensional case only BB is able to solve
this problem, even for small $n$. 

Comparing the optimal subsets of seven different point constructions, our key findings are that 
(1) the discrepancy of the best size-$m$ subset can be significantly better than the original size-$m$ construction (with the only exception of the {\tt Fibon} sequence with $m\ge 80$) and the main improvement stems from increasing $n$ from $m$ to $m+20$,  
(2) the values of the best found subsets are very similar for all low-discrepancy constructions, regardless of their comparatively large differences in the original $m=n$ constructions; 
(3) {\tt unif} and {\tt iLHS} point sets are not competitive in $3d$ in terms of discrepancy values. 

Given that many computer science applications operate with a fixed budget problem dimension $d$ and a fixed budget $n$ of points that can be evaluated, we consider it valuable to collect point sets of small discrepancy values. Our work shows that the subset selection approach could be an interesting alternative to construct such point sets. However, while the point sets identified in this work are better than the best ones found using the classic approaches, we still need to compare with explicit star discrepancy minimization  (i.e., using classic optimization approaches to identify the points $p_1, \ldots, p_m$ which minimize the star discrepancy) and with other low-discrepancy sequences suggested in the literature, such as the symmetrized Fibonacci sequences suggested in~\cite{BilykTY12}, generalized Halton sequences~\cite{Braaten1979}, etc. 

Another very important next step for our work would be the design of efficient algorithms to address the discrepancy subset selection problem for larger point sets $m>100$ and for dimensions $d>3$. As discussed in Section~\ref{sec:comp-algos}, we do not expect MILP formulations and BB approaches and variants to generalize well. Improving the quality of the upper and lower bounds for BB, or the use of other techniques such as column generation or branch and cut/price, should allow for better performance only on slightly larger instances. Given the computational complexity of the star discrepancy evaluation, we can expect that it is impossible to find algorithms that scale polynomially in $n$ and $d$. Heuristic solutions, tailored to the star discrepancy settings such as the ``snapping'' procedures in~\cite{GnewuchWW12} may therefore be needed. 

\subsection*{Acknowledgments} 
We thank Magnus Wahlstr\"om for providing his implementation of the algorithm from~\cite{discrepancy_algorithm} to evaluate the star  discrepancy. We also thank Michael Gnewuch and Aicke Hinrichs for suggesting to add Fibonacci sequences to our experiments in $2d$, and we thank Gon\c{c}alo Martins for 
performing preliminary experiments on branch and bound~\cite{Goncalo-master}. Last but not least, we thank the reviewers for their constructive feedback and for pointing us to the work of Dwivedi et al. on online thinning~\cite{Thinning1}. 

This work was funded by the  Paris Ile-de-France Region and by national funds through the FCT - Foundation for Science and Technology, I.P. within the scope of the project CISUC -- UID/CEC/00326/2020. 



\begin{thebibliography}{DFGGR19}

\bibitem[AH14]{randomUpperbounds}
Christoph Aistleitner and Markus Hofer.
\newblock Probabilistic discrepancy bound for {M}onte {C}arlo point sets.
\newblock {\em Math. Comp.}, 83:1373--1381, 2014.

\bibitem[BG02]{iLHS}
B.~Beachkofski and R.~Grandhi.
\newblock Improved {Distributed} {Hypercube} {Sampling}.
\newblock In {\em 43rd {AIAA}/{ASME}/{ASCE}/{AHS}/{ASC} {Structures},
  {Structural} {Dynamics}, and {Materials} {Conference}}, pages 1--7. American
  Institute of Aeronautics and Astronautics, 2002.

\bibitem[BGK{\etalchar{+}}17]{BousquetGKTV17}
Olivier Bousquet, Sylvain Gelly, Karol Kurach, Olivier Teytaud, and Damien
  Vincent.
\newblock Critical hyper-parameters: No random, no cry.
\newblock {\em CoRR}, abs/1706.03200, 2017.

\bibitem[BLV08]{BilykLV08}
Dmitriy Bilyk, Michael~T. Lacey, and Armen Vagharshakyan.
\newblock On the small ball inequality in all dimensions.
\newblock {\em Journal of Functional Analysis}, 254:2470--2502, 2008.

\bibitem[BTY12]{BilykTY12}
Dmitriy Bilyk, Vladimir~N. Temlyakov, and Rui Yu.
\newblock Fibonacci sets and symmetrization in discrepancy theory.
\newblock {\em J. Complex.}, 28(1):18--36, 2012.

\bibitem[BW79]{Braaten1979}
E.~Braaten and G.~Weller.
\newblock An improved low-discrepancy sequence for multidimensional
  quasi-{M}onte {C}arlo integration.
\newblock {\em J. of Comput. Phys.}, 33(2):249--258, 1979.

\bibitem[BZ93]{BZ93}
P.~Bundschuh and Y.~C. Zhu.
\newblock A method for exact calculation of the discrepancy of low-dimensional
  point sets {I}.
\newblock {\em Abh. Math. Sem. Univ. Hamburg}, 63:115--133, 1993.

\bibitem[CCD{\etalchar{+}}20]{CauwetCDLRRTTU20}
Marie{-}Liesse Cauwet, Camille Couprie, Julien Dehos, Pauline Luc,
  J{\'{e}}r{\'{e}}my Rapin, Morgane Rivi{\`{e}}re, Fabien Teytaud, Olivier
  Teytaud, and Nicolas Usunier.
\newblock Fully parallel hyperparameter search: Reshaped space-filling.
\newblock In {\em Proc. of the 37th International Conference on Machine
  Learning, {ICML}}, volume 119 of {\em Proceedings of Machine Learning
  Research}, pages 1338--1348. {PMLR}, 2020.

\bibitem[DDG18]{DoerrDG18LHSbounds}
Benjamin Doerr, Carola Doerr, and Michael Gnewuch.
\newblock Probabilistic lower bounds for the discrepancy of {L}atin hypercube
  samples.
\newblock In Josef Dick, Frances~Y. Kuo, and Henryk Wo{\'{z}}niakowski,
  editors, {\em Contemporary Computational Mathematics - A Celebration of the
  80th Birthday of Ian Sloan}, pages 339--350. Springer, 2018.

\bibitem[DEM96]{discrepancy_algorithm}
David~P. Dobkin, David Eppstein, and Don~P. Mitchell.
\newblock Computing the discrepancy with applications to supersampling
  patterns.
\newblock {\em {ACM} Trans. Graph.}, 15(4):354--376, 1996.

\bibitem[DFGGR19]{Thinning1}
Raaz Dwivedi, Ohad N. Feldheim, Ori Gurel-Gurevich, and Aaditya Ramdas.
\newblock The power of online thinning in reducing discrepancy.
\newblock {\em Probability Theory and Related Fields}, 174:103--131, 2019.

\bibitem[DGW14]{DoerrGW14}
Carola Doerr, Michael Gnewuch, and Magnus Wahlstr{\"o}m.
\newblock Calculation of discrepancy measures and applications.
\newblock In William Chen, Anand Srivastav, and Giancarlo Travaglini, editors,
  {\em A Panorama of Discrepancy Theory}, pages 621--678. Springer, 2014.

\bibitem[Doe14]{Doerr14lowerBoundRandomPoints}
Benjamin Doerr.
\newblock A lower bound for the discrepancy of a random point set.
\newblock {\em J. Complex.}, 30(1):16--20, 2014.

\bibitem[DP10]{DickP10}
J.~Dick and F.~Pillichshammer.
\newblock {\em Digital Nets and Sequences}.
\newblock Cambridge University Press, Cambrigde, 2010.

\bibitem[DR13]{evolutionary_extension}
Carola Doerr and Fran{\c{c}}ois{-}Michel~De Rainville.
\newblock Constructing low star discrepancy point sets with genetic algorithms.
\newblock In {\em Proc. of Genetic and Evolutionary Computation Conference
  (GECCO'13)}, pages 789--796. ACM, 2013.

\bibitem[Fau82]{Faure82}
H.~Faure.
\newblock Discrepancy of sequences associated with a number system (in
  dimension s).
\newblock {\em Acta. Arith}, 41(4):337--351, 1982.
\newblock In French.

\bibitem[GH21]{GnewuchH20LHSlupper}
Michael Gnewuch and Nils Hebbinghaus.
\newblock Discrepancy bounds for a class of negatively dependent random points
  including {L}atin hypercube samples.
\newblock {\em Annals of Applied Probability}, 2021.
\newblock To appear. Available at
  \url{https://imstat.org/journals-and-publications/annals-of-applied-probability/annals-of-applied-probability-future-papers/}.

\bibitem[GJ90]{Domi}
Michael~R. Garey and David~S. Johnson.
\newblock {\em Computers and Intractability; A Guide to the Theory of
  NP-Completeness}.
\newblock W. H. Freeman \& Co., USA, 1990.

\bibitem[GJ97]{GalantiJ97OPtionPricing}
S.~Galanti and A.~Jung.
\newblock Low-discrepancy sequences: Monte {C}arlo simulation of option prices.
\newblock {\em Journal of Derivatives}, pages 63--83, 1997.

\bibitem[GKWW12]{w1_discrepancy}
Panos Giannopoulos, Christian Knauer, Magnus Wahlstr{\"{o}}m, and Daniel
  Werner.
\newblock Hardness of discrepancy computation and {\(\epsilon\)}-net
  verification in high dimension.
\newblock {\em J. Complexity}, 28(2):162--176, 2012.

\bibitem[GSW09]{complexity}
Michael Gnewuch, Anand Srivastav, and Carola Winzen.
\newblock Finding optimal volume subintervals with k points and calculating the
  star discrepancy are {NP}-hard problems.
\newblock {\em J. Complexity}, 25(2):115--127, 2009.

\bibitem[GWW12]{GnewuchWW12}
Michael Gnewuch, Magnus Wahlstr{\"{o}}m, and Carola Winzen.
\newblock A new randomized algorithm to approximate the star discrepancy based
  on threshold accepting.
\newblock {\em {SIAM} J. Numerical Analysis}, 50(2):781--807, 2012.

\bibitem[Hal64]{Halton64}
J.~H. Halton.
\newblock {Algorithm 247: Radical-Inverse Quasi-random Point Sequence}.
\newblock {\em {Communications of the ACM}}, 7(12):701~--~702, 1964.

\bibitem[Ham60]{Ham60}
J.~Hammersley.
\newblock Monte {C}arlo methods for solving multivariable problems.
\newblock {\em Annals of the New York Academy of Sciences}, 86, 1960.

\bibitem[Hin04]{Hin04}
Aicke Hinrichs.
\newblock Covering numbers, {V}apnik-\v{C}ervonenkis classes and bounds for the
  star-discrepancy.
\newblock {\em J. Complexity}, 20:477--483, 2004.

\bibitem[HNWW01]{HNWW01}
S.~Heinrich, E.~Novak, G.~W. Wasilkowski, and H.~Wo\'{z}niakowski.
\newblock The inverse of the star-discrepancy depends linearly on the
  dimension.
\newblock {\em Acta Arith.}, 96:279--302, 2001.

\bibitem[JK08]{JK08SobolGneeration}
S.~Joe and F.~Y. Kuo.
\newblock Constructing {S}obol' sequences with better two-dimensional
  projections.
\newblock {\em SIAM J.~Sci. Comput.}, 30:2635--2654, 2008.

\bibitem[Mar19]{Goncalo-master}
Gon{\c c}alo Corte-Real Martins.
\newblock Algorithms for the star discrepancy subset selection problem.
\newblock Master thesis. University of Coimbra. Available at
  \url{https://estudogeral.sib.uc.pt/bitstream/10316/87919/1/Thesis_v2.pdf},
  2019.

\bibitem[Mat09]{Mat99}
J.~Matou\v{s}ek.
\newblock {\em Geometric Discrepancy}.
\newblock Springer, Berlin, 2nd edition, 2009.

\bibitem[MBC79]{LHS}
Michael~D. McKay, Richard~J. Beckman, and William~J. Conover.
\newblock {A Comparison of Three Methods for Selecting Values of Input
  Variables in the Analysis of Output from a Computer Code}.
\newblock {\em Technometrics}, 21:239--245, 1979.

\bibitem[NGD{\etalchar{+}}18]{NeumannGDN018}
Aneta Neumann, Wanru Gao, Carola Doerr, Frank Neumann, and Markus Wagner.
\newblock Discrepancy-based evolutionary diversity optimization.
\newblock In {\em Proc. of Genetic and Evolutionary Computation Conference
  (GECCO'18)}, pages 991--998. ACM, 2018.

\bibitem[Nie72]{Nie72a}
H.~Niederreiter.
\newblock Discrepancy and convex programming.
\newblock {\em Ann. Mat. Pura Appl.}, 93:89--97, 1972.

\bibitem[Nie92]{Nie92}
H.~Niederreiter.
\newblock {\em Random Number Generation and Quasi-{M}onte {C}arlo Methods},
  volume~63 of {\em SIAM CBMS-NSF Regional Conference Series in Applied
  Mathematics}.
\newblock SIAM, Philadelphia, 1992.

\bibitem[Sch72]{SchmidtLowerD2}
W.~M. Schmidt.
\newblock Irregularities of distribution vii.
\newblock {\em Acta. Arith}, 21:45--50, 1972.

\bibitem[Sob67]{Sobol}
I.~M. Sobol.
\newblock {On the Distribution of Points in a Cube and the Approximate
  Evaluation of Integrals}.
\newblock {\em {USSR Computational Mathematics and Mathematical Physics}},
  7(4):86~--~112, 1967.

\bibitem[SWN03]{santner_design_2003}
Thomas~J. Santner, Brian~J. Williams, and William~I. Notz.
\newblock {\em The {Design} and {Analysis} of {Computer} {Experiments}}.
\newblock Springer {Series} in {Statistics}. Springer, 2003.

\bibitem[VC06]{RevHal}
Bart Vandewoestyne and Ronald Cools.
\newblock Good permutations for deterministic scrambled {H}alton sequences in
  terms of ${L}_2$-discrepancy.
\newblock {\em Journal of Computational and Applied Mathematics}, 189:341 --
  361, 2006.

\bibitem[Whi77]{White77}
B.~E. White.
\newblock On optimal extreme-discrepancy point sets in the square.
\newblock {\em Numer. Math.}, 27:157--164, 1976/77.

\end{thebibliography}
\newcommand{\etalchar}[1]{$^{#1}$}

\appendix
\section{Computational results}
The tables on the following pages present the details of the observations summarized in the main body of the paper. 

\begin{table}[hb!]
\centering
   \footnotesize
	\begin{tabular}{rr*{2}{r@{ }r@{ }r@{ }r@{  \quad }}}\toprule
                \multirow{2}{*}{$m$}& \multirow{2}{*}{$n$}& 
		\multicolumn{4}{c}{{\tt iLHS} sequences} & \multicolumn{4}{c}{{\tt unif} sequences}\\
        &       &$\min$  & med & $\max$ & $(succ)$ &$\min$  & med & $\max$ & $(succ)$ 
	     \\\midrule

        $20$ & $40$&0   &161  &1742& (9) &1   &31   &464 &(9)\\
             & $60$&86  &620  &1305& (8) &140 &305  &569 &(7) \\
             & $80$&567 &567  &567 & (1) &966 &966  &966 & (1) \\
	     &$100$&-   &-    &-   & (0) &- &-  &- & (0) \\\midrule
        $40$ & $60$&253 &872  &1492& (2) &- & - & - & (0)\\
             & $80$&-   &-    &-   & (0) &-  &-   &-   & (0)\\
             &$100$&-   &-    &-   & (0) &-  &-   &-   & (0)\\\midrule
        $60$ & $80$&92  &806  &1519& (2) &25  &86   &147  & (2)\\
             &$100$&-   &-    &-   & (0) &-  &-   &-  & (0)\\\midrule
	$80$ &$100$&-   &-    &-   & (0) &.  &-   &-  & (0)\\\bottomrule
\end{tabular}
        \caption{CPU-time (minimum, median, maximum and number of successful runs out of 10) 
 	taken by branch and bound, for several values of $n$
and $m$ for {\tt iLHS} and {\tt unif} point sets
in the $3d$ case, where ``-''
indicates that the approach did not terminate before the time limit of 1800
seconds.} %
\label{CPU-3d-lhs}
\end{table}

\begin{table}
\centering
        {\footnotesize
                \begin{tabular}{l@{  }ll*{6}{r}|}\toprule
                        \multirow{1}{*}{$m$} & \multirow{1}{*}{sequence} 
                          & \multicolumn{1}{c}{$n=40$}  & \multicolumn{1}{c}{$n=60$} 
                            & \multicolumn{1}{c}{$n=80$}  & \multicolumn{1}{c}{$n=100$} 
                            & \multicolumn{1}{c}{$n=120$} & \multicolumn{1}{c}{$n=140$}\\ \midrule

$20$ & {\tt Faure}     &1.3556 &1.6926 &2.5039 &2.4591 &2.4024&4.3098  \\
     & {\tt Sobol}     &1.3556 &2.2080 &2.3290 &2.7040 &3.6873&4.6088  \\
     & {\tt Halton}    &1.2828 &1.5105 &2.8167 &2.7912 &2.7148&3.6817  \\
     & {\tt RevHal}  &1.6044 &1.9179 &2.7750 &2.7750 &2.7111&4.9240  \\
     & {\tt Fibon}     &2.1678 &2.9110 &3.6353 &4.4520 &5.1599&5.8107  \\\midrule
$40$ &  {\tt Faure}     & &1.1263 &1.5811 &1.5055 &1.4535 & *2.6538\\
     & {\tt Sobol}     & &1.4267 &1.4695 &1.7000 &*2.2972& *2.7916\\
     & {\tt Halton}    & &1.0000 &1.7417 &1.7000 &*1.6681& *2.2615\\
     & {\tt RevHal}  & &1.2107 &1.7750 &1.6889 &1.6458 & *3.0579\\
     & {\tt Fibon}     & &1.8513 &2.2650 &2.7402 &3.1938 & *3.5181\\\midrule
$60$ &  {\tt Faure}     & & &1.1971 &1.1570 & *1.1285 &*2.0154\\
     & {\tt Sobol}     & & &1.1304 &1.2942 & *1.6949 &*2.1389  \\
     & {\tt Halton}    & & &1.3167 &1.2833 & *1.2750 &*1.7231  \\
     & {\tt RevHal}  & & &1.3439 &1.3063 & *1.4222 &*2.2702  \\
     & {\tt Fibon}     & & &1.6114 &2.0355 & *2.4153 &*2.7430 \\\midrule
$80$ &  {\tt Faure}     & & & &1.0543 &*1.0016 &1.6392 \\
     & {\tt Sobol}     & & & &1.1004 &*1.4218 &*1.5409 \\
     & {\tt Halton}    & & & &1.0718 &*1.0194 &*1.4077 \\
     & {\tt RevHal}  & & & &1.1250 &*1.0222 &*1.8906 \\
     & {\tt Fibon}     & & & &1.6066 &*1.9283 &*2.2117 \\\midrule
$100$ &  {\tt Faure}    & & & & &1.0000 &1.4063 \\
      & {\tt Sobol}    & & & & &1.2469 &*1.2034 \\
      & {\tt Halton}   & & & & &1.0000 &*1.2434 \\
      & {\tt RevHal} & & & & &1.0000 &*1.5956\\
      & {\tt Fibon}    & & & & &1.6097 &*1.8532 \\\midrule
$120$ &  {\tt Faure}    & & & & & &1.2055 \\
      & {\tt Sobol}    & & & & & &1.0000 \\
      & {\tt Halton}   & & & & & &1.1077 \\
      & {\tt RevHal} & & & & & &1.3565 \\
      & {\tt Fibon}    & & & & & &1.5534\\\bottomrule
        \end{tabular}}
        \caption{Integrality gap of the LP relaxation of MILP for two-dimensional
        deterministic sequences with respect to the optimal found and, when 
        not available, with respect to the best solution found (marked with ``*'').}
\label{tab:2dLP}
\end{table}

\begin{table}[b]
\centering
        {\footnotesize
                \begin{tabular}{l@{  }l*{6}{r@{ }r@{ \quad }}}\hline
                \multirow{2}{*}{$m$} & \multirow{2}{*}{sequence}
                            & \multicolumn{2}{c}{$n=40$}  & \multicolumn{2}{c}{$n=60$} 
                            & \multicolumn{2}{c}{$n=80$}  & \multicolumn{2}{c}{$n=100$} 
                            & \multicolumn{2}{c}{$n=120$} & \multicolumn{2}{c}{$n=140$}\\
                     &      & MILP & BB & MILP & BB & MILP & BB & MILP & BB & MILP & BB & MILP & BB \\\toprule
$20$ & {\tt Faure}          &4&0  &25 &0 &82  &2  &967 &  11 &-   & 19  &-   &83  \\
   & {\tt Sobol}            &5&0  &19 &0 &114 &2  &740 &   5 &-   & 15  &-   &60  \\  
   & {\tt Halton}           &6&0  &26 &0 &89  &2  &-   &  14 &-   & 44  &-   &72  \\ 
   & {\tt RevHal}           &4&0  &28 &0 &87  &1  &-   &  10 &-   & 32  &-   &102  \\
   & {\tt Fibon}            &4&0  &45 &0 &-   &2  &-   &  26 &-   &20   &-   &67\\\midrule 

$40$ & {\tt Faure}       &   &   &43 &1 &169 &24 &-   & 215 &-   & 1795&-   &- \\
   & {\tt Sobol}         &   &   &19 &0 &479 &45 &-   & 374 &-   & -   &-   &- \\  
   & {\tt Halton}        &   &   &9  &8 &294 &13 &-   & 216 &-   & -   &-   &- \\ 
   & {\tt RevHal}        &   &   &10 &1 &214 &15 &1560&216  &-   & 1235&-   &- \\
   & {\tt Fibon}         &   &   &53 &0 &284 &18 &-   & 187 &-   & -   &-   &- \\\hline   

$60$ & {\tt Faure}       &   &   &   &      &33 &14 &-& 496 &-&-&-&- \\
   & {\tt Sobol}         &   &   &   &      &41 &56 &-& 200 &-&-&-&- \\  
   & {\tt Halton}        &   &   &   &      &135&8  &-&1106 &-&-&-&- \\ 
   & {\tt RevHal}        &   &   &   &      &47 &10 &-& 761 &-&-&-&- \\
   & {\tt Fibon}         &   &   &   &      &143&7  &-& 518 &-&-&-&- \\\midrule 

$80$ & {\tt Faure}       &   &   &   &      &   &   &161&1194 &-&-&-&- \\
   &  {\tt Sobol}        &   &   &   &      &   &   &254& 123 &-&-&-&- \\  
   & {\tt Halton}        &   &   &   &      &   &   &843& 161 &-&-&-&- \\ 
   & {\tt RevHal}        &   &   &   &      &   &   &517& 305 &-&-&-&- \\
   & {\tt Fibon}         &   &   &   &      &   &   &1608& 45 &-&-&-&-  \\\hline 

$100$ & {\tt Faure}      &   &   &   &      &   &   &    &    &19  &-&-&- \\
   &  {\tt Sobol}        &   &   &   &      &   &   &    &    &1538&847&-&- \\  
   & {\tt Halton}        &   &   &   &      &   &   &    &    &12  & - &-&- \\ 
   & {\tt RevHal}        &   &   &   &      &   &   &    &    &12  & - &-&- \\
   & {\tt Fibon}         &   &   &   &      &   &   &    &    &-   &104&-&- \\\midrule
$120$ & {\tt Faure}      &   &   &   &      &   &   &    &    &    &   &321&-  \\
   &  {\tt Sobol}        &   &   &   &      &   &   &    &    &    &   &127&-  \\  
   &{\tt Halton}         &   &   &   &      &   &   &    &    &    &   &1332&-  \\ 
   & {\tt RevHal}        &   &   &   &      &   &   &    &    &    &   &491 &1253  \\
   & {\tt Fibon}         &   &   &   &      &   &   &    &    &    &   &-&915 \\\bottomrule

        \end{tabular}}
        \caption{CPU-time in seconds obtained by the ILP solver on the MILP formulation and by BB, for several values of $n$
and $m$ for $2d$ low-discrepancy sequences, where ``-'' indicates that the approach did not terminate before the time limit.}
\label{tab:CPU-2d}
\end{table}

\begin{table}
        {\footnotesize
		\begin{tabular}{rr
			r@{ }r@{ }r@{ }r@{  }r@{ }r@{ }r@{ }r@{ \quad }
			r@{ }r@{ }r@{ }r@{  }r@{ }r@{ }r@{ }r@{ }}\toprule
                && 
		\multicolumn{8}{c}{{\tt iLHS} sequences} & \multicolumn{8}{c}{{\tt unif} sequences}\\\midrule
	\multirow{2}{*}{$m$}& \multirow{2}{*}{$n$}& \multicolumn{4}{c}{MILP} & 
	   \multicolumn{4}{c}{BB} & \multicolumn{4}{c}{MILP} & \multicolumn{4}{c}{BB} \\   
        &    &$\min$  & med & $\max$ & $(succ)$ 
             &$\min$  & med & $\max$ & $(succ)$ 
             &$\min$  & med & $\max$ & $(succ)$ 
             &$\min$  & med & $\max$ & $(succ)$ 
	     \\\midrule

         $20$ & $40$&2   &2    &4   &(10) &0   &0    &0   &(10)
                    &0   &1    &1   &(10) &0   &0    &11  &(10)\\
               & $60$&18  &30   &45  &(10) &0   &0    &0   &(10)
                     &1   &16   &41  &(10) &0   &2    &201 &(10)\\
               & $80$&71  &140  &190 & (9) &1   &2    &5   &(10) 
                     &15  &86   &458 &(10) &2   &12   &134 &(10)\\ 
               &$100$&689 &913  &1323& (7) &8   &13   &231 &(10) 
                     &281 &1122 &1375& (6) &4   &21   &230 &(10)\\ 
               &$120$&1486&1486 &1486& (1) &15  &30   &49  &(10) 
                     &1343&1555 &1768& (2) &21  &91   &252 &(10)\\ 
               &$140$&-   &-    &-   & (0) &75  &107  &185 &(10) 
                     &-   &-    &-   & (0) &35  &210  &348 &(10)\\\hline 
 $40$ & $60$&5   &17   &26  & (10)&0   &0    &3   &(10)  
                 &0   &2    &4   & (10)&2   &192  &922 & (6)\\
             & $80$&101 &160  &247 & (10)&5   &12   &19  &(10)  
                   &3   &6    &20  & (10)&81  &94   &107 & (2)\\ 
             &$100$&-   &-    &-   & (0) &64  &143  &265 &(10)  
                   &11  &111  &1045& (9) &3   &615  &1227& (2)\\ 
             &$120$&-   &-    &-   & (0) &413 &1089 &1488& (7)  
                   &7   &376  &1424& (8) &185 &185  &185 & (1)\\ 
             &$140$&-   &-    &-   & (0) &-   &-    &-   & (0) 
                   &8   &8    &8   & (1) &678 &678  &678 & (1)\\\hline 

        $60$ & $80$&15  &31   &62  & (10)&3   &7    &104   &(10)  
                   &2   &3    &51  & (10)&0   &287  &1691 & (3)\\ 
             &$100$&355 &501  &1707& (10)&110 &195  &643   &(10)  
                   &3   &9    &66  & (10)&-  &-  &- & (0)\\ 
             &$120$&-   &-    &-   & (0) &1030&1220 &1411  &(2)  
                   &4   &95   &500 & (10)&-  &-  &- & (0)\\ 
             &$140$&-   &-    &-   & (0) &-  &-    &-     &(0) 
                   &76  &227  &1349& (8) &-  &-  &- & (0)\\\hline 
        
	$80$ &$100$&75  &127  &273 & (10)&22  &59   &652 &(10)  
                   &3   &4    &7   & (10)&0   &3    &4   & (3)\\ 
             &$120$&292 &881  &1459& (6) &759 &1423 &1774&(4) 
                   &4   &14   &233 & (10)&-   &-    &-   & (0)\\ 
             &$140$&-   &-    &-   & (0) &-   &-    &-   & (0)   
                   &7   &28   &191 & (9) &-   &-    &-   & (0)\\ \hline

       $100$ &$120$&23  &167  &324 & (10) &24  &120  &365&(6)  
                   &3   &5    &18  & (10) &5   &5    &5   &(1)\\
             &$140$&565 &1017 &1703& (3)  &-   &-    &-   &(0)  
                   &6   &10   &18  & (10)&27   &27   &27  & (1)\\\hline 
        $120$ &$140$&35  &259  &528 & (10) &796 &853  &1114&(4) 
                   &5   &8    &19  & (10)&-   &-    &-   & (0)\\\hline 
\end{tabular}}
        \caption{CPU-time (minimum, median, maximum and number of successful runs out of 10) 
        of the ILP solver on the MILP formulation, and BB, for several values of $n$ and $m$ for {\tt iLHS} and {\tt unif} point sets
in the $2d$ case, where ``-'' indicates
that the approach did not terminate before the time limit of 1800 seconds.}
\label{CPU-2d-lhs}
\end{table}

\begin{landscape}
\begin{table}
	{\scriptsize
		\begin{tabular}{l@{  }ll*{6}{r@{  \ }r@{ \ }r}}\toprule
			\multirow{2}{*}{$m$} & \multirow{2}{*}{seq.} & \multirow{2}{*}{$n=m$}
			  & \multicolumn{3}{c}{$n=40$}  & \multicolumn{3}{c}{$n=60$} 
			    & \multicolumn{3}{c}{$n=80$}  & \multicolumn{3}{c}{$n=100$} 
			    & \multicolumn{3}{c}{$n=120$} & \multicolumn{3}{c}{$n=140$}\\
		         & && $rand.$ & $greed.$ & $subset$ & $rand.$ & $greed.$ & $subset$  
			    & $rand.$ & $greed.$ & $subset$ & $rand.$ & $greed.$ & $subset$  
			    & $rand.$ & $greed.$ & $subset$ & $rand.$ & $greed.$ & $subset$ \\\midrule
$20$ & \tt Faure     &0.2094 &0.0911&0.1169&\underline{0.0834} &0.0994&0.1656&\underline{0.0785}&0.0922&0.1305&0.0776& 0.1036&0.1305&0.0762&0.1018&0.1554&0.0745&0.1038&0.1554&0.0738  \\
     & \tt Sobol'     &0.1313 &0.0938&0.1254&\underline{0.0834}&0.0960&0.1656&0.0809 &0.1000&0.1472&0.0785&0.1014&0.1472&\underline{0.0743}&0.1029&0.1554&0.0743&0.1031&0.1905&0.0738  \\  
     & \tt Halton    &0.1477 &0.0944&0.1681&0.0861&0.1000&0.1463&0.0833&0.0979&0.1537&0.0782&0.0965&0.1537&0.0775&0.1025&0.1315&0.0754&0.1029&0.1315&0.0739  \\ 
     &\tt  RevHal &0.1500 &0.0935&0.1375&0.0836&0.0977&0.1639&0.0829&0.1000&0.1241&\underline{0.0771}&0.1012&0.1421&0.0771&0.1059&0.1335&0.0753&0.1039&0.1481&0.0736  \\ 
     & \tt Fibon     &0.0930 &0.0931&0.1257&0.0866&0.0971&0.1390&0.0828&0.1000&0.1598&0.0790&0.1023&0.1779&0.0757&0.1009&0.1506&\underline{0.0741}&0.1038&0.1449&\underline{ \bf 0.0731}  \\\midrule

$40$ & \tt Faure      &0.0836 &   & &   &0.0590&0.0747&0.0523 &0.0656&0.0832&0.0490 &0.0691&0.0945&\underline{0.0467}&0.0722&0.1154&\underline{0.0451}&0.0714&0.1154&*0.0454 \\
     & \tt Sobol'      &0.0836 &   & &   &0.0613&0.0695&0.0522&0.0666&0.0899&0.0495 &0.0656&0.0899&\underline{0.0467}&0.0697&0.1154&*0.0463&0.0687&0.1117&\underline{\bf*0.0447} \\  
     & \tt Halton     &0.0993 &   & &   &0.0611&0.1162&0.0552&0.0656&0.0972&\underline{0.0484}&0.0696&0.1157&0.0472&0.0729&0.1157&*0.0463&0.0719&0.1157&*0.0454 \\ 
     & \tt RevHal  &0.0866 &   & &   &0.0628&0.0880&0.0523 &0.0667&0.0841&0.0493&0.0699&0.0841&0.0469&0.0703&0.1021&0.0457&0.0730&0.0985&*0.0457 \\   
     & \tt Fibon      &0.0545 &   & &   &0.0583&0.0726&\underline{0.0498}&0.0656&0.1008&0.0485&0.0669&0.0807&0.0475&0.0713&0.0862&0.0463&0.0714&0.0742 & *0.0449  \\\midrule
     
$60$ & \tt Faure      &0.0645 &   & &   &    &   &  &0.0464&0.0705&0.0371 &0.0522&0.0726&0.0359&0.0540&0.0872&*0.0350&0.0564&0.0937&*0.0345 \\
     & \tt Sobol'      &0.0484 &   & &   &    &   &  &0.0472&0.0659&0.0381&0.0510&0.0807&0.0356&0.0540&0.0820&*\underline{0.0341}&0.0562&0.0917&*0.0343 \\  
     & \tt Halton     &0.0654 &   & &   &    &   &  &0.0453&0.0644 &0.0366&0.0516&0.0625&0.0357&0.0539&0.0604&*0.0354&0.0561&0.0609&*0.0346 \\ 
     & \tt RevHal  &0.0626 &   & &   &    &   &  &0.0468&0.0646&0.0391 &0.0500&0.0583&0.0363&0.0530&0.0667&*0.0354&0.0561&0.0661&*\underline{\bf 0.0339} \\   
     & \tt Fibon      &0.0363& & &   & & &   &0.0436&0.0826&\underline{0.0364}&0.0498&0.0925&\underline{0.0345}&0.0531&0.0866&*0.0345&0.0537&0.0720&*0.0344\\\midrule 

$80$ & \tt Faure      &0.0452 &   & &   &    &   &  & &    &   &0.0397&0.0472&0.0327 &0.0432&0.0468&*0.0311&0.0448&0.0433&0.0281 \\
     & \tt Sobol'      &0.0506 &   & &   &    &   &  & &    &   &0.0398&0.0432&0.0302 &0.0424&0.0439&*0.0286&0.0461&0.0546&*0.0285 \\  
     & \tt Halton     &0.0375 &   & &   &    &   &  & &    &   &0.0387&0.0500&0.0298&0.0424&0.0454&*0.0283&0.0435&0.0454&*0.0282 \\ 
     & \tt RevHal  &0.0454 &   & &   &    &   &  & &    &   &0.0387&0.0426&0.0313 &0.0430&0.0556&*0.0284&0.0439&0.0521&*0.0282 \\
     & \tt Fibon     &{\bf 0.0272} & & &   & & &   & & &  &0.0349&0.0674 &\underline{0.0282} &0.0395&0.0683&\underline{*0.0280}   &0.0432&0.0551 &\underline{*0.0279} \\\midrule

$100$ & \tt Faure     &0.0461 &   & &   &    &   &  & &    &   & &    &     &0.0340&0.0471&0.0310&0.0348&0.0386&0.0241 \\
      & \tt Sobol'     &0.0398 &   & &   &    &   &  & &    &   & &    &     &0.0323&0.0471&0.0262&0.0365&0.0471&*0.0253\\  
      & \tt Halton    &0.0502 &   & &   &    &   &  & &    &   & &    &     &0.0329&0.0432&0.0299&0.0364&0.0532&*0.0250 \\ 
      & \tt RevHal &0.0416 &   & &   &    &   &  & &    &   & &    &     &0.0330&0.0488&0.0299&0.0334&0.0566&*0.0238 \\   
     & \tt Fibon     &0.0232& & &   & & &   & & &   & & &  &0.0298 &0.0463 &\underline{\bf 0.0230}   &0.0349 &0.0531 &\underline{*0.0232}  \\\midrule

$120$ & \tt Faure     &0.0372 &   & &   &    &   &  & &    &   & &    &     & &    &     &0.0273&0.0332&0.0211  \\
      & \tt Sobol'     &0.0251 &   & &   &    &   &  & &    &   & &    &     & &    &     &0.0277&0.0329&0.0227  \\  
      & \tt Halton    &0.0423 &   & &   &    &   &  & &    &   & &    &     & &    &     &0.0292&0.0323&0.0222  \\ 
      & \tt RevHal &0.0417 &   & &   &    &   &  & &    &   & &    &     & &    &     &0.0279&0.0298&0.0213  \\  
     & \tt Fibon      &0.0210 & & &   & & &   & & &   & & &   & & &   &0.0254 &0.0379 &\underline{ \bf 0.0199}  \\\bottomrule

	\end{tabular}}
	\caption{$2d$, low-discrepancy sequences: best found star discrepancy values found by $random$ subset sampling, by the $greedy$ heuristic, compared to the optimal or the best found (marked with *) values returned by MILP or BB (column $subset$), for all tested combinations of $n$ and $m$. 
	Best star discrepancy values for each ($n$,$m$) combination are \underline{underlined}, and the minimum for each $m$ is highlighted in \textbf{boldface}.}
\label{tab:2dsummary}
\end{table}

\begin{table}
        {\footnotesize
		\begin{tabular}{rrrrr*{2}{rrr}r@{ }lrrr}\toprule
                \multirow{2}{*}{$m$}& 
                \multicolumn{3}{c}{$m = n$} & \multirow{2}{*}{$n$}& 
                \multicolumn{3}{c}{$random$} & \multicolumn{3}{c}{$greedy$} & 
                \multicolumn{5}{c}{$subset$} \\           
                     &$\min$  & med & $\max$ & 
                     &$\min$  & med & $\max$ 
                     &$\min$  & med & $\max$   
		     &$\min$ & & med & $\max$ & $(succ)$\\\midrule
		     20   &0.1022&0.1403 &0.1714&  
                    40 &0.0892& 0.0936& 0.0956& 0.1201& 0.1406& 0.1647 
		       &0.0838& ($M$,$B$) & 0.0866& 0.0894& (10)\\
              &&&&  60 &0.0953& 0.0969& 0.0992& 0.1186& 0.1364& 0.1580
	               &0.0786& ($M$,$B$) & 0.0817& 0.0825& (10)\\
              &&&&  80 &0.0942& 0.0993& 0.1009& 0.1226& 0.1465& 0.2022
	               &0.0774& ($M$,$B$) & 0.0785& 0.0800& (10)\\
              &&&& 100 &0.0970& 0.1022& 0.1035& 0.1256& 0.1493& 0.1897
	      	       &0.0756& ($B$) & 0.0770& 0.0778& (10) \\
              &&&& 120 &0.0993& 0.1042& 0.1071& 0.1146& 0.1471& 0.1952
	      	       &0.0734& ($B$) & 0.0748& 0.0762&  (10)\\
              &&&& 140 &0.1001& 0.1034& 0.1070& 0.1293& 0.1394& 0.1876
	      	       &\bf 0.0731 & ($B$) & 0.0742& 0.0747& (10)\\\midrule
         40   &0.0682&0.0779&0.1004&  
                    60 &0.0588& 0.0621& 0.0640& 0.0763& 0.0889& 0.0994
                       &0.0507& ($M$,$B$) & 0.0520& 0.0556&  (10)\\
              &&&&  80 &0.0601& 0.0658& 0.0676& 0.0693& 0.0900& 0.1284
                       &0.0478& ($M$,$B$) & 0.0486& 0.0496&  (10)\\
              &&&& 100 &0.0663& 0.0683& 0.0711& 0.0793& 0.0975& 0.1264
                       &0.0465& ($M$,$B$) &0.0470& 0.0477& (10)\\
              &&&& 120 &0.0696& 0.0708& 0.0724& 0.0743& 0.0905& 0.1209
                       &0.0452& ($B$) &0.0460& 0.0465&  (7)\\
              &&&& 140 &0.0719& 0.0725& 0.0743& 0.0864& 0.1041& 0.1323
                       &\bf *0.0445& ($B$) &0.0450& 0.0539&  (1)\\\midrule
         60   &0.0508&0.0619&0.0680&  
                   80 &0.0450& 0.0466& 0.0481& 0.0562& 0.0665& 0.0882
                       &0.0380&($M$,$B$) & 0.0396& 0.0427& (10)\\
              &&&& 100 &0.0498& 0.0504& 0.0527& 0.0606& 0.0736& 0.0919
                       &0.0356&($M$,$B$) & 0.0363& 0.0368&  (10)\\
              &&&& 120 &0.0529& 0.0541& 0.0562& 0.0555& 0.0795& 0.0974
                       &*0.0343&($M$,$B$) & 0.0348& 0.0354&   (2)\\
              &&&& 140 &0.0543& 0.0565& 0.0581& 0.0673& 0.0760& 0.1053
                       &\bf *0.0338&($B$) & 0.0345& 0.0350&  (0)\\\midrule
         80   &0.0389&0.0447&0.0567& 
                   100 &0.0376& 0.0386& 0.0391& 0.0460& 0.0530& 0.0616
                       &0.0304&($M$,$B$) & 0.0317& 0.0340& (10)\\
              &&&& 120 &0.0408& 0.0425& 0.0434& 0.0474& 0.0578& 0.0876
                       &0.0288&($M$) & 0.0291& 0.0297&  (7) \\
              &&&& 140 &0.0430& 0.0449& 0.0467& 0.0512& 0.0615& 0.0715
                       &\bf *0.0280&($B$) & 0.0287& 0.0300&  (0)\\\midrule
        100  &0.0301&0.0388&0.04833& 
                   120 &0.0322& 0.0335& 0.0341& 0.0380& 0.0433& 0.0479
                       &0.0259&($M$,$B$) & 0.0270& 0.0290& (10)\\
              &&&& 140 &0.0355& 0.0366& 0.0378& 0.0395& 0.0460& 0.0604
                       &\bf 0.0246& ($M$) &0.0252& 0.0257&   (3)\\\midrule
         120  &0.0280&0.0368&0.0436& 
                   140 &0.0279& 0.0294& 0.0304& 0.0314& 0.0383& 0.0462
                       &\bf 0.0229& ($M$,$B$) &0.0240&  0.0270 & (10) \\\bottomrule
 \end{tabular}} \caption{$2d$, {\tt iLHS}:  
 Minimum, median, and maximum 
 of the best
 star discrepancy values found for ten independently generated {\tt iLHS} point sets 
 per each combination of $m$ and $n$ in $d=2$. 
 We show values returned by $random$, by $greedy$, and by the exact strategies (column $subset$). 
 Where none of MILP or BB converged within the given time frame of 1800 seconds, 
 the best found upper bound is shown (marked by a *); the number in parenthesis counts the point sets for which the optimal value could be computed by at least one of the two exact solvers. 
 The best value found for a given instance is the minimum
 obtained by all exact approaches. The minimum value for
 each $m$ is in \textbf{boldface}.  
}
\label{tab:2dsummary_lhs}
\end{table}

\begin{table}
        {\footnotesize
		\begin{tabular}{rrrrr*{2}{rrr}r@{ }lrrr}\toprule
                \multirow{2}{*}{$m$}& 
                \multicolumn{3}{c}{$m = n$} & \multirow{2}{*}{$n$}& 
                \multicolumn{3}{c}{$random$} & \multicolumn{3}{c}{$greedy$} & 
                \multicolumn{5}{c}{$subset$} \\           
                     &$\min$  & med & $\max$ & 
                     &$\min$  & med & $\max$ 
                     &$\min$  & med & $\max$   
		     &$\min$ & & med & $\max$ & $(succ)$\\\midrule 
  20  & 0.1836 & 0.2773 & 0.3450 &
 	  40 & 0.1014 & 0.1143 & 0.1272 & 0.1403 & 0.1621 & 0.1943
	     & 0.0992 & ($M$,$B$) & 0.1139 & 0.1272 & (10)\\
     &&&& 60 & 0.1026 & 0.1081 & 0.1398 & 0.1288 & 0.1546 & 0.2334
     	     & 0.0865 & ($M$,$B$) & 0.0956 & 0.1398 & (10)\\
     &&&& 80 & 0.1001 & 0.1053 & 0.1141 & 0.1405 & 0.1542 & 0.2080
     	     & 0.0821 & ($M$,$B$) & 0.0854 & 0.0925 & (10)\\
     &&&&100 & 0.0999 & 0.1049 & 0.1086 & 0.1232 & 0.1540 & 0.1947
             & 0.0787 & ($M$,$B$) & 0.0803 & 0.0833 & (10)\\
     &&&&120 & 0.1020 & 0.1055 & 0.1114 & 0.1366 & 0.1519 & 0.2259
	     & 0.0765 & ($B$) & 0.0778 & 0.0798 & (10) \\
     &&&&140 & 0.1014 & 0.1083 & 0.1107 & 0.1218 & 0.1472 & 0.1777
     	     & \bf 0.0753 & ($B$) & 0.0763 & 0.0775 &  (10)\\\midrule
  40  & 0.1369 &  0.1648 &  0.2625 &
          60 & 0.0771 & 0.0856 & 0.1077 & 0.0973 & 0.1293 & 0.1685
     	     & 0.0678 & ($M$,$B$) & 0.0830 & 0.1077 &  (10)\\
     &&&& 80 & 0.0757 & 0.0801 & 0.0938 & 0.0931 & 0.1094 & 0.1534
     	     & 0.0624 & ($M$,$B$) &  0.0704 & 0.0938 &  (10)\\
     &&&&100 & 0.0757 & 0.0805 & 0.0846 & 0.0923 & 0.1133 & 0.1765
     	     & 0.0531 &  ($M$,$B$) & 0.0586 & 0.0643 & (9)\\
     &&&&120 & 0.0753 & 0.0791 & 0.0835 & 0.0753 & 0.1048 & 0.1431
     	     & 0.0498 &  ($M$) & 0.0552 & 0.0737 &  (8)\\
     &&&&140 & \bf 0.0749 & 0.0792 & 0.0846 & 0.0836 & 0.1105 & 0.1221  
             & \bf *0.0479 &  ($B$) & 0.0494 & 0.0666 & (2)\\\midrule
 60  & 0.1205 & 0.1583 & 0.2123 &
	  80  & 0.0678 & 0.0750 & 0.1016 & 0.0838 & 0.1011 & 0.1231
 	     & 0.0629 &  ($M$,$B$) & 0.0745 & 0.1016 &  (10)\\
     &&&&100 & 0.0661 & 0.0712 & 0.1047 & 0.0790 & 0.1009 & 0.1474
     	     & 0.0496 &  ($M$) & 0.0580 & 0.1047 &  (10)\\
     &&&&120 & 0.0637 & 0.0683 & 0.0882 & 0.0754 & 0.0901 & 0.1114
     	     & 0.0422 &  ($M$) & 0.0527 & 0.0882 &  (10)\\
     &&&&140 & 0.0623 & 0.0687 & 0.0713 & 0.0849 & 0.0993 & 0.1244
 	     & \bf 0.0404 &  ($M$) & 0.0442 & 0.0498 &  (8)\\\midrule
 80  & 0.0913& 0.1343 & 0.2101 &
         100 & 0.0569 & 0.0651 & 0.0853 & 0.0628 & 0.0825 & 0.1047
     	     & 0.0569 &  ($M$) & 0.0632 & 0.0853 &  (10)\\
     &&&&120 & 0.0565 & 0.0617 & 0.0836 & 0.0696 & 0.0817 & 0.1236
     	     & 0.0427 &  ($M$) & 0.0502 & 0.0787 &  (10)\\
     &&&&140 & 0.0565 & 0.0624 & 0.0645 & 0.0645 & 0.0789 & 0.1124
 	     & \bf 0.0399 &  ($M$) & 0.0450 & 0.0577 &  (9)\\\midrule
 100 & 0.0946 & 0.1172 & 0.1774 &
	 120 & 0.0547 & 0.0706 & 0.1030 & 0.0732 & 0.0954 & 0.1316
 	     & \bf 0.0452 &  ($M$) & 0.0706 & 0.1030 &  (10)\\
     &&&&140 & 0.0546 & 0.0595 & 0.0748 & 0.0701 & 0.0756 & 0.1081
 	     &  0.0481 &  ($M$) & 0.0509 & 0.0748 & (10)\\\midrule
 120 & 0.0662&  0.1116& 0.1511 &
	 140 & \bf 0.0557 & 0.0604 & 0.0692 & 0.0580 & 0.0721 & 0.0890
	     & \bf 0.0557 &  ($M$) & 0.0588 & 0.0692 &  (10)\\\bottomrule
\end{tabular}} \caption{$2d$, {\tt unif}: 
 Minimum, median, and maximum of the best star discrepancy values found for ten independently sampled {\tt unif} point sets 
 per each combination of $m$ and $n$ in $d=2$. 
 We show values returned by $random$, by $greedy$, and by the exact strategies (column $subset$). 
 Where none of MILP or BB converged within the given time frame of 1800 seconds, 
 the best found upper bound is shown (marked by a *); the number in parenthesis counts the point sets for which the optimal value could be computed by at least one of the two exact solvers.  
 The best value found for a given instance is the minimum
 obtained by all exact approaches for any of the ten point sets. The minimum value for
 each $m$ is in \textbf{boldface}.  
 }
\label{tab:2dsummary_unif}
\end{table}

\begin{table}
	{\footnotesize
		\begin{tabular}{l@{  }ll*{4}{r@{  \ }r@{ \ }r}}\toprule
			\multirow{2}{*}{$m$} & \multirow{2}{*}{sequence} & \multirow{2}{*}{$n=m$}
			  & \multicolumn{3}{c}{$n=40$}  & \multicolumn{3}{c}{$n=60$} 
			    & \multicolumn{3}{c}{$n=80$}  & \multicolumn{3}{c}{$n=100$}\\
		         & && $random$ & $greedy$ & $subset$ & $random$ & $greedy$ & $subset$  
			    & $random$ & $greedy$ & $subset$ & $random$ & $greedy$ & $subset$ \\\midrule
$20$ & \tt Faure 	  &0.1795 &  0.1559&0.2206&0.1316 &0.1612&0.2518&\underline{0.1205}  &0.1664&0.3428&*0.1223 &0.1714&0.3193&*0.1225 \\
     & \tt Sobol' 	  &0.1774 &  0.1493&0.1758&\underline{0.1268}&0.1616&0.1817&0.1220  &0.1590&0.2028&*\underline{\bf 0.1202}&0.1692&0.2028&*0.1263 \\
     & \tt Halton 	  &0.2079 &  0.1635&0.1800&0.1311 &0.1664&0.1917&*0.1240 &0.1663&0.1912&*0.1214 &0.1648&0.1990&*0.1244 \\
     & \tt RevHal  &0.1870 &  0.1511&0.1767&0.1300 &0.1509&0.1956&0.1250  &0.1635&0.1633&0.1207 &0.1678&0.1801&*\underline{0.1242} \\\midrule
     
$40$ & \tt Faure 	  &0.1836 &  && &0.1007&0.1239&*\underline{0.0805}&0.1024&0.1654&*\underline{\bf 0.0778} &0.1073&0.1781&*\underline{0.0801} \\
     & \tt Sobol' 	  &0.1066 &  && &0.1039&0.1172&*0.0817 &0.1114&0.1323&*0.0810 &0.1144&0.1311&*0.0786 \\
     & \tt Halton 	  &0.1475 &  && &0.1028&0.1425& 0.0854 &0.1015&0.1464&*0.0809 &0.1083&0.1470&*0.0807 \\
     & \tt RevHal  &0.1333 &  && &0.1091&0.1441& 0.0817 &0.1091&0.1441&*0.0799 &0.1119&0.14371&*0.0812 \\\midrule
     
$60$ & \tt Faure 	  &0.1107 &  && &&& &0.0744&0.0900&*\underline{\bf 0.0606} &0.0857&0.0874&*0.0666 \\
     & \tt Sobol' 	  &0.0736 &  && &&& &0.0834&0.0892&*0.0674 &0.0875&0.0948&*\underline{0.0643}\\
     & \tt Halton 	  &0.1081 &  && &&& &0.0775&0.0917&*0.0642 &0.0842&0.0878&*0.0648 \\
     & \tt RevHal  &0.0866 &  && &&& &0.0778&0.0883&*0.0654 &0.0839&0.0882&*0.0648 \\\midrule

$80$ & \tt Faure 	  &0.0640 &  && &&& &&& &0.0661&0.0846&*0.0638 \\
     & \tt Sobol' 	  &0.0828 &  && &&& &&& &0.0655&0.0637&*0.0605 \\
     & \tt Halton 	  &0.0700 &  && &&& &&& &0.0640&0.0618&*0.0550 \\
     & \tt RevHal  &0.0747 &  && &&& &&& &0.0644&0.0792&*\underline{\bf 0.0547} \\\bottomrule
	\end{tabular}}
	\caption{$3d$, low-discrepancy sequences: best found star discrepancy values found by $random$ subset sampling, by the $greedy$ heuristic, compared to the optimal or the best found (marked with *) values returned by MILP or BB (column $subset$), for all tested combinations of $n$ and $m$. 
	Best star discrepancy values for each ($n$,$m$) combination are \underline{underlined}, and the minimum for each $m$ is highlighted in \textbf{boldface}.
	}
	\label{tab:3dsummary}
\end{table}
\end{landscape}

\begin{landscape}

\begin{table}
    \footnotesize{
		\begin{tabular}{rrrrr*{3}{rrr}r}\toprule
                \multirow{2}{*}{$m$}& 
                \multicolumn{3}{c}{$m = n$} & \multirow{2}{*}{$n$}& 
                \multicolumn{3}{c}{$random$} & \multicolumn{3}{c}{$greedy$} & 
                \multicolumn{4}{c}{$subset$} \\           
                     &$\min$  & med & $\max$ & 
                     &$\min$  & med & $\max$  
                     &$\min$  & med & $\max$   
		     &$\min$  & med & $\max$ & $(succ)$\\\midrule 
		20   & 0.3946 & 0.4325& 0.4721 &  
      40 & 0.2008&  0.2158& 0.2759& 0.1936& 0.2261& 0.2459& 0.1462& 0.1678& 0.2184&(9)\\
&&&&  60 & 0.1911&  0.2184& 0.2421& 0.1917& 0.2132& 0.2654& 0.1388& 0.1503& 0.1933&(7)\\
&&&&  80 & 0.1958&  0.2081& 0.2333& 0.1665& 0.1918& 0.2455& \bf *0.1290& 0.1368& 0.1553&(1)\\
&&&& 100 & 0.1974&  0.2093& 0.2221& 0.1754& 0.1984& 0.2143& *0.1296& 0.1383& 0.1976& (0)\\\midrule
	40 & 0.3003& 0.3364& 0.4545&
      60 & 0.1990&  0.2441& 0.2739& 0.1618& 0.1993& 0.2303& *0.1515& 0.1752& 0.2302& (2)\\
&&&&  80 & 0.1896&  0.2139& 0.2395& 0.1348& 0.1543& 0.1675& *0.1026& 0.1211& 0.1577& (0)\\ 
&&&& 100 & 0.1955&  0.2013& 0.2180& 0.1283& 0.1427& 0.1667& *\bf 0.0977& 0.1333& 0.1600& (0)\\\midrule
	60 & 0.3352& 0.3673& 0.4173& 
      80 & 0.2314& 0.2583& 0.2747& 0.1639& 0.1943& 0.2336& 0.1639& 0.1906& 0.2336&(2) \\ 
&&&& 100 & 0.1945& 0.2211& 0.2537& 0.1061& 0.1192& 0.1547& *\bf 0.1056& 0.1157& 0.1547& (0)\\  \midrule
	80 & 0.2851& 0.3630& 0.3792&
     100 & 0.2087& 0.2640& 0.2795& \bf 0.1709& 0.2143& 0.2305& *\bf 0.1709& 0.2143& 0.2305 & (0)\\\bottomrule

 \end{tabular}} \caption{
 $3d$, {\tt iLHS}: 
 Minimum, median, and maximum of the best star discrepancy values found for ten independently sampled {\tt iLHS} point sets 
 per each combination of $m$ and $n$ in $d=3$. 
 We show values returned by $random$, by $greedy$, and by the exact strategies (column $subset$). 
 Where none of MILP or BB converged within the given time frame of 1800 seconds, 
 the best found upper bound is shown (marked by a *); the number in parenthesis counts the point sets for which the optimal value could be computed by at least one of the two exact solvers. 
 The best value found for a given instance is the minimum
 obtained by all exact approaches for any of the ten point sets. The minimum value for
 each $m$ is in \textbf{boldface}.
 }
\label{tab:3dsummary_lhs}
\end{table}

\begin{table}
\footnotesize{
		\begin{tabular}{rrrrr*{3}{rrr}r}\toprule
                \multirow{2}{*}{$m$}& 
                \multicolumn{3}{c}{$m = n$} & \multirow{2}{*}{$n$}& 
                \multicolumn{3}{c}{$random$} & \multicolumn{3}{c}{$greedy$} & 
                \multicolumn{4}{c}{$subset$} \\           
                     &$\min$  & med & $\max$ & 
                     &$\min$  & med & $\max$  
                     &$\min$  & med & $\max$   
		     &$\min$  & med & $\max$ & $(succ)$\\\midrule 
		20   &0.4089& 0.4923& 0.5606&
	40 & 0.2063& 0.2423& 0.2802& 0.1952& 0.2272& 0.2603& 0.1693& 0.2017& 0.2603& (9) \\ 
  &&&&  60 & 0.1955& 0.2109& 0.2442& 0.1657& 0.2112& 0.2487& 0.1363& 0.1452& 0.1598& (7)\\ 
  &&&&  80 & 0.1977& 0.2126& 0.2219& 0.1686& 0.2101& 0.2975& *0.1338& 0.1406& 0.1813& (1)\\
  &&&& 100 & 0.1999& 0.2091& 0.2272& 0.1755& 0.1966& 0.2398& *\bf 0.1292& 0.1415& 0.1470& (0)\\ \midrule
	40 & 0.3375 & 0.3889& 0.4524&
        60 & 0.2100& 0.2659& 0.2921& 0.1850& 0.2358& 0.2578& *0.1499& 0.2358& 0.2578& (0)\\ 
  &&&&  80 & 0.2136& 0.2219& 0.2645& 0.1279& 0.1558& 0.2601& *0.1168& 0.1369& 0.1748& (0)\\
  &&&& 100 & 0.1898& 0.2209& 0.2504& 0.1419& 0.1712& 0.2150& *\bf 0.1127& 0.1404& 0.1691& (0)\\\midrule
	60 & 0.3378& 0.3992& 0.4491& 
	80 & 0.2354& 0.2533& 0.3181& 0.1855& 0.2367& 0.3181& *0.1855& 0.2367& 0.3181& (0)\\ 
  &&&& 100 & 0.2324& 0.2599& 0.2953& 0.1234& 0.1660& 0.2625& *\bf 0.1187& 0.1614& 0.2161& (0)\\ \midrule
	80 & 0.3190& 0.3510& 0.4197&
	100 & 0.2279& 0.2795& 0.3176& \bf 0.1872& 0.2410& 0.2801& *\bf 0.1872& 0.2410& 0.2801& (0)\\\bottomrule

 \end{tabular}} \caption{
$3d$, {\tt unif}: 
 Minimum, median, and maximum of the best star discrepancy values found for ten independently sampled {\tt unif} point sets 
 per each combination of $m$ and $n$ in $d=2$. 
 We show values returned by $random$, by $greedy$, and by the exact strategies (column $subset$). 
 Where none of MILP or BB converged within the given time frame of 1800 seconds, 
 the best found upper bound is shown (marked by a *); the number in parenthesis counts the point sets for which the optimal value could be computed by at least one of the two exact solvers. 
 The best value found for a given instance is the minimum
 obtained by all exact approaches for any of the ten point sets. The minimum value for
 each $m$ is in \textbf{boldface}.
 }
\label{tab:3dsummary_unif}
\end{table}

\end{landscape}

\end{document}